\documentclass[12pt]{iopart}

\usepackage{cite}
\usepackage{iopams}  
  \expandafter\let\csname equation*\endcsname\relax
  \expandafter\let\csname endequation*\endcsname\relax

\usepackage{amsmath,amsthm}
\usepackage{graphicx}
\usepackage{ wasysym }

\usepackage{stmaryrd}

\usepackage{color}

\newcommand{\ol}[1]{\overline{#1}}
\newcommand{\ub}{\ol{u}}

\newcommand{\rb}{\mathbf{r}}

\newtheorem{theorem}{Theorem}
\theoremstyle{definition}
\newtheorem{mydef}{Definition}

\begin{document}

\title[Discontinuous shock solutions of the Whitham modulation equations]{Discontinuous shock solutions of the Whitham modulation equations as dispersionless limits of traveling waves}

\author{Patrick Sprenger and Mark A. Hoefer}
\address{University of Colorado Boulder, Boulder, CO, USA}
\ead{patrick.sprenger@colorado.edu and hoefer@colorado.edu}
\vspace{10pt}
\begin{indented}
\item[]July 2019
\end{indented}

\begin{abstract}
  Whitham modulation theory describes the zero dispersion limit of
  nonlinear waves by a system of conservation laws for the parameters
  of modulated periodic traveling waves. Here, admissible,
  discontinuous, weak solutions of the Whitham modulation
  equations---termed \emph{Whitham shocks}---are identified with zero
  dispersion limits of traveling wave solutions to higher order
  dispersive partial differential equations (PDEs). The far-field
  behavior of the traveling wave solutions satisfies the
  Rankine-Hugoniot jump conditions.  Generally, the numerically
  computed traveling waves represent heteroclinic connections between
  two periodic orbits of an ordinary differential equation.  The focus
  here is on the fifth order Korteweg-de Vries equation. The three
  admissible one-parameter families of Whitham shocks are used as
  solution components for the generalized Riemann problem of the
  Whitham modulation equations. Admissible KdV5-Whitham shocks are
  generally undercompressive, i.e., all characteristic families pass
  through the shock.  The heteroclinic traveling waves that limit to
  admissible Whitham shocks are found to be ubiquitous in numerical
  simulations of smoothed step initial conditions for other higher
  order dispersive equations including the Kawahara equation (with
  third and fifth order dispersion) and a nonlocal model of weakly
  nonlinear gravity-capillary waves with full dispersion. Whitham
  shocks are linked to recent studies of nonlinear higher order
  dispersive waves in optics and ultracold atomic
  gases. 
  The approach presented here provides a novel method for constructing
  new traveling wave solutions to dispersive nonlinear wave equations
  and a framework to identify physically relevant, admissible shock
  solutions of the Whitham modulation equations.
\end{abstract}

%
%
%
%
%

\section{Introduction}
Nonlinear dispersive wave equations are host to numerous solutions
ranging from steady periodic traveling waves and solitary waves to
unsteady dispersive shock waves (DSWs). The canonical model equation
that includes third order (long wave) dispersion is the Korteweg-de
Vries (KdV) equation
\begin{equation}
  \label{eq:KdV}
  u_t + uu_x + u_{xxx} = 0 .
\end{equation}
In their seminal paper, Gurevich and Pitaevskii (GP)
\cite{gurevich_nonstationary_1974} construct the KdV DSW by solving a
Riemann problem consisting of step initial conditions for the
KdV-Whitham modulation equations. The KdV-Whitham modulation equations
are a first order, quasi-linear system of three conservation laws for
the slow evolution of a KdV nonlinear periodic traveling wave's
parameters \cite{whitham_non-linear_1965}. Despite the apparent
contradiction of applying the Whitham modulation equations---obtained
by averaging conservation laws of the KdV equation under the
assumption of \textit{long} wavelength modulations---to discontinuous
initial data, the GP DSW modulation solution is not a discontinuous
shock but rather consists of a self-similar, centered rarefaction wave
solution (a 2-wave \cite{el_dispersive_2016-1}) between two constant
levels. In a series of papers, Lax, Levermore
\cite{lax_zero_1979,lax_small_1983-1,lax_small_1983-2,lax_small_1983}
and Venakedis
\cite{venakides_zero-dispersion_1985,venakides_korteweg-vries_1990}
utilized the inverse scattering transform to prove that the
KdV-Whitham system (and their multiphase generalizations
\cite{flaschka_multiphase_1980}) describe the zero dispersion limit of
the KdV equation. The limit was shown to be weak in the $L_2$ sense
and thus provides a rigorous justification for Whitham's approach.

Because the KdV-Whitham equations are strictly hyperbolic
\cite{levermore_hyperbolic_1988}, a generic class of initial data will
necessarily lead to singularity formation and multivalued solution
regions.  Singularity formation in the KdV-Whitham equations is
regularized by adding an additional phase to the multiphase modulation
solution.  The simplest case is breaking of the dispersionless Hopf
equation (the zero-phase KdV-Whitham equation) for the mean flow that
is regularized by the addition of an oscillatory region that emanates
from a point of gradient catastrophe in the $x$-$t$ plane and is
described by the one-phase KdV-Whitham equations for modulated
periodic waves
\cite{gurevich_breaking_1991,gurevich_evolution_1992,tian_initial_1994}.
The boundaries of this modulated, one-phase region matches
continuously to the zero-phase solution, which generalizes the GP step
problem to smooth initial conditions.  Singularity formation in the
one-phase Whitham equations similarly leads to $x$-$t$ regions that
are described by the two-phase Whitham equations
\cite{flaschka_multiphase_1980} and match continuously to the solution
of the one-phase Whitham equations \cite{grava_generation_2002}.
Generally, compressive waves in the KdV-Whitham equations that lead to
gradient catastrophe are regularized in terms of global, continuous
expansion waves.  This is also the case for other integrable systems
such as the defocusing nonlinear Schr\"{o}dinger equation
\cite{el_general_1995} so that one can regularize singularity
formation in the Whitham equations by an appropriate choice of
degenerate initial conditions for sufficiently large phase modulations
that yield global, continuous solutions
\cite{bloch_dispersive_1992,kodama_whitham_1999,biondini_whitham_2006}.

Despite our understanding of singularity formation in the Whitham
equations for integrable systems, there remains an outstanding
question regarding discontinuous shock solutions.  Do discontinuous,
weak solutions of the Whitham equations have any meaning vis-a-vis the
governing partial differential equation (PDE) and the physical system
that it models?  One standard approach in the conservation law
community is to identify discontinuous shock solutions of hyperbolic
systems as \textit{admissible} if they are the pointwise limit of
traveling wave (TW) solutions to the dissipatively or
dissipative-dispersively perturbed conservation law as dictated by the
microscopic physics of the problem
\cite{lefloch_hyperbolic_2002,dafermos_hyperbolic_2009}.  The TW
profiles are heteroclinic, equilibrium-to-equilibrium solutions of a
stationary ordinary differential equation (ODE).  This approach leads
to both admissible compressive Lax shocks
\cite{lax_hyperbolic_1973,whitham_linear_1974} and admissible
nonclassical undercompressive shocks \cite{jacobs_travelling_1995}.
But the Whitham modulation equations are the zero dispersion limit of
a conservative PDE so that there is no justification for introducing
dissipative perturbations to the Whitham equations.

Instead, we expand the collection of TWs and prove that if TW
solutions of a conservative PDE consist of heteroclinic connections
between an equilibrium or a periodic orbit and another periodic orbit
exist, their far-field behavior satisfies the Rankine-Hugoniot
relations for discontinuous, shock solutions of the Whitham modulation
equations.  
This enables us to define admissible shock solutions of the Whitham
modulation equations---\textit{Whitham shocks}---as (weak) zero
dispersion limits of oscillatory, heteroclinic TW solutions. The
solution of the governing PDE corresponding to a Whitham shock
consists of two disparate oscillatory waves that are connected by a
transition region occurring over the length scale of a single
oscillation, the dispersive coherence length
\cite{el_dispersive_2016-1}.

A phase plane analysis shows that the KdV equation \eqref{eq:KdV} does
not admit heteroclinic TW solutions \cite{el_dispersive_2017}.
Consequently, the Whitham shocks studied here are not admissible
solutions of the KdV-Whitham equations.  Instead, we consider a
different class of nonlinear dispersive equations.
 
While KdV is a universal model of DSWs in convex media
\cite{el_dispersive_2016-1}, modifications to it are required when
higher order, e.g., short wavelength, or large amplitude, effects
occur. The KdV equation is then modified to include higher order
dispersive and/or nonlinear terms that more accurately model the
underlying physics. This manuscript focuses on the implications of
higher order and nonlocal dispersive effects on both steady, traveling
wave solutions as well as unsteady DSWs. A canonical model
incorporating higher order dispersion and weak nonlinearity is the
fifth order Korteweg-de Vries equation (KdV5)
\begin{equation}\label{eq:kdv5}
  u_t + uu_x + u_{xxxxx} = 0.
\end{equation}
The dispersionless limit of the KdV5 equation \eqref{eq:kdv5} which
will be of use throughout this manuscript is made explicit via the
hydrodynamic variable transformation $X = \epsilon x$,
$T = \epsilon t$, and $U(X,T;\epsilon) = u(X/\epsilon,T/\epsilon)$.
Substitution of these variables into \eqref{eq:kdv5} yields
\begin{equation}\label{eq:kdv5_dispersionless}
  U_T + UU_X + \epsilon^4 U_{XXXXX} = 0,
\end{equation}
so that the singular dispersionless limit is achieved by passing
$\epsilon \to 0$.

In a typical multi-scale asymptotic expansion for weakly nonlinear,
long waves, the dispersion relation expansion for small wavenumber
involves successively smaller terms. In this scenario, cubic
dispersion typically dominates, which is the reason the KdV equation
is said to be universal. However, when successive terms are
comparable, e.g., third and fifth order dispersion near an inflection
point \cite{sprenger_shock_2017}, then higher order dispersion is
operable. The KdV5 equation \eqref{eq:kdv5} is an example in which
fifth order dispersion dominates over all others, but other well known
examples also exist. A long wave, nonlinear model that incorporates
the competition between third and fifth order dispersion is the
Kawahara equation \cite{kawahara_oscillatory_1972}
\begin{equation}\label{eq:kawahara}
  u_t + uu_x + \alpha u_{xxx} + u_{xxxxx} = 0,
\end{equation}
where $\alpha \in \mathbb{R}$ is a parameter. Higher order dispersive
effects can also occur when full, or nonlocal dispersion is
included. A class of models in this vein is
\begin{equation}\label{eq:whitham_eq_intro}
  u_t + uu_x + \mathcal{K} * u_x = 0, \quad \mathcal{K}(x) = \frac{1}{2\pi}\int_{-\infty}^{\infty} \frac{\omega(k)}{k } e^{i k  x} dk, 
\end{equation}
where $\omega(k)$ is the linear dispersion relation and 
\begin{align}\label{eq:whit_kernel}
  (K * u_x)(x,t) = \int_\mathbb{R}K(x-y)u_y(y,t) dy.
\end{align}
This model was originally proposed by Whitham as a weakly nonlinear,
fully dispersive model of water waves \cite{whitham_variational_1967,
  whitham_linear_1974} and has since been called the Whitham equation,
not to be mistaken for the aforementioned Whitham modulation
equations. Since then, this model has been used in applications to
understand the propagation of water waves with surface tension
\cite{dinvay_whitham_2017}. The dispersion relation for
gravity-capillary water waves is
\begin{align*}
  \omega_{\rm ww}(k) &= \sqrt{ k(1 + B k^2) \tanh k}
  \\ &\sim k +\frac{3B - 1}{6}k^3 + \frac{(19-15 B (3 B+2))}{360}k^5 +
       \ldots, \quad 0 < k \ll 1
\end{align*}
so that the Fourier transform \eqref{eq:whitham_eq_intro} and
convolution \eqref{eq:whit_kernel} should be considered in the
distributional sense
\cite{ehrnstrom_traveling_2009,ehrnstrom_existence_2012}.  For Bond
number $B$ (the ratio of surface tension to gravity forces) near
$1/3$, higher order dispersion is important. Generally, whenever
linear dispersion exhibits an inflection point for nonzero wavenumber
$k$, higher order dispersive effects take precedence
\cite{sprenger_shock_2017}.

Dispersive systems with higher order dispersion occur in a variety of
applications, for example, gravity-capillary water waves
\cite{hunter_existence_1988,dias_fully-nonlinear_2010,matsuno_hamiltonian_2015}
and flexural ice sheets
\cite{marchenko_long_1988,guyenne_forced_2014,dinvay_whitham_2019},
nonlinear optical systems
\cite{wai_nonlinear_1986,webb_generalized_2013,smyth_dispersive_2016,el_radiating_2016},
and Bose-Einstein condensates \cite{khamehchi_negative-mass_2017} all
exhibit a change in dispersion curvature for sufficiently short
waves. Nonlinear PDEs with higher order or nonlocal dispersion admit
solutions wholly different from their lower order dispersive
counterparts. Examples include multi-mode periodic waves (traveling
waves in which two wavenumbers are resonant)
\cite{trichtchenko_instability_2016,gao_asymmetric_2016,remonato_numerical_2017},
solitary waves consisting of multiple pulses
\cite{champneys_homoclinic_1998}, solitary waves with decaying
oscillatory tails
\cite{kawahara_oscillatory_1972,grimshaw_solitary_1994}, and solitary
waves accompanied by co-propagating small, but finite, amplitude
oscillations spatially extending to infinity
\cite{hunter_solitary_1983,akylas_solitary_1992,clamond_plethora_2015}.
This rich zoology of traveling waves (TWs) can be attributed to
additional degrees of freedom inherent in the higher order ODE that
results from a TW ansatz. In this paper, we identify new TW solutions
to further expand this zoology.

The increased variety of solutions to higher order dispersive PDEs is
not limited to steady solutions. The unsteady dynamics that result
from dispersively regularized gradient catastrophe, most simply
embodied by the step initial data
\begin{align}
  \label{eq:riemann_data}
  u(x,0) = \begin{cases} \Delta & x < 0\\ 0 & x \geq 0 \end{cases} ,
                                              \quad \Delta > 0
\end{align} 
for, e.g., equations \eqref{eq:kdv5} and \eqref{eq:kawahara}, are
notably different from lower order dispersive systems such as KdV
\eqref{eq:KdV}.  Previous numerical and asymptotic studies have
jointly identified the following distinct, unsteady regimes for
Riemann problems of higher order or nonlocal PDEs:
\begin{enumerate}
\item DSW implosion: above a critical jump height, the DSW's harmonic
  wave edge becomes modulationally unstable and a multiphase structure
  emerges \cite{lowman_dispersive_2013,maiden_modulations_2016}.
\item Radiating DSW (RDSW): a perturbed convex DSW with exponentially
  small radiation emitted from the solitary wave edge. This results
  from a linear resonance between the solitary wave and shorter linear
  waves \cite{el_radiating_2016,sprenger_shock_2017}.
\item Traveling DSW (TDSW): a rapid transition from an equilibrium
  level to a comoving nearly periodic wavetrain slowly transitions to
  a constant level via a partial DSW from the intermediate periodic
  wavetrain to a background constant
  \cite{conforti_resonant_2014,sprenger_shock_2017,hoefer_modulation_2019}. See
  Figure \ref{fig:TDSW_intro} for an example TDSW solution of the KdV5
  equation.
\item Crossover DSW: can be thought of as a combination of the RDSW
  and TDSW, where the linear resonance begins to grow to a finite
  amplitude wavetrain and the wave dynamics are more complex
  \cite{el_radiating_2016,sprenger_shock_2017}.
\end{enumerate}
The original motivation for this manuscript was to fully describe the
TDSW solution in Fig.~\ref{fig:TDSW_intro} in terms of modulation
theory.  As we show in Sec.~\ref{sec:applications}, the TDSW can be
described in terms of a shock-rarefaction solution to the KdV5-Whitham
and the Kawahara-Whitham modulation equations.  However, during our
research, a broader theme emerged and revealed a host of novel
heteroclinic and homoclinic TW solutions as well as a Whitham
modulation theory framework to interpret them.


\begin{figure}
\begin{center}\label{fig:TDSW_intro}
\includegraphics[scale=0.45]{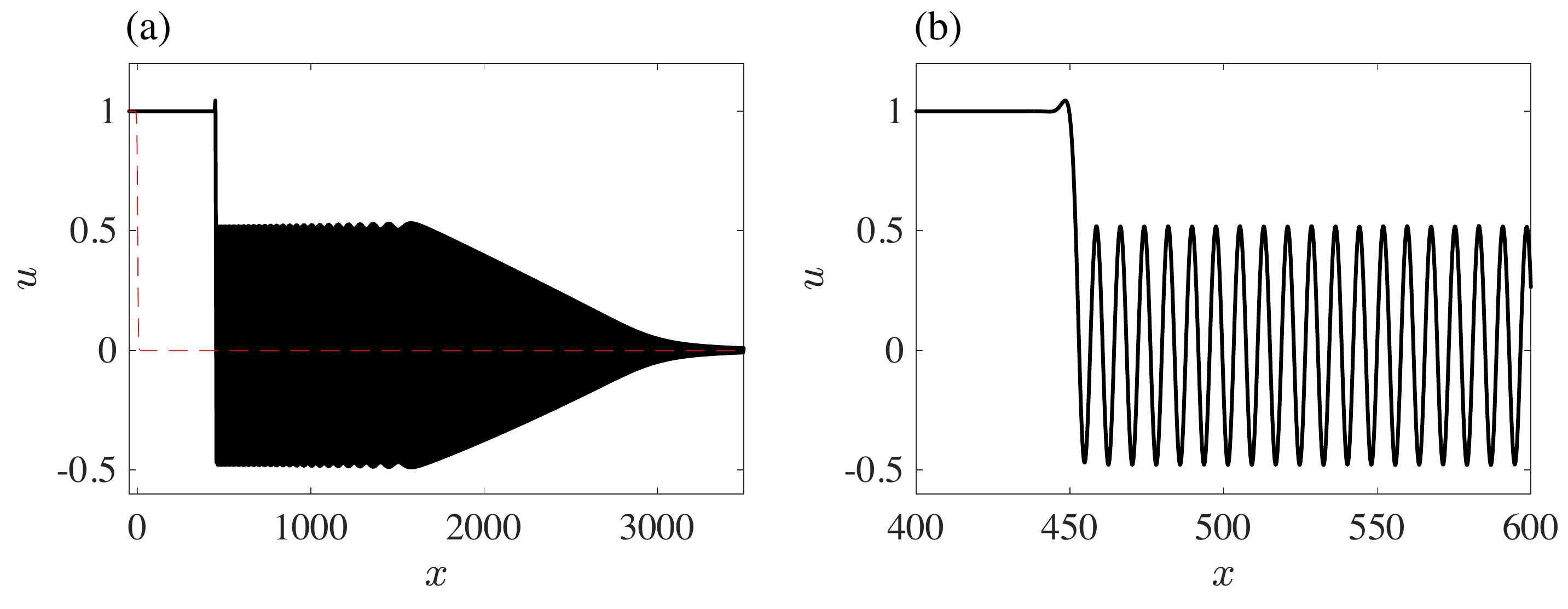}
\caption{Numerical simulation of the Riemann problem
  \eqref{eq:riemann_data} with $\Delta = 1$ for the KdV5 equation
  \eqref{eq:kdv5}. (a) Solution at $t = 1000$ (solid) for a TDSW and
  the initial step (dashed). (b) A zoom in of the trailing edge that
  resembles an equilibrium heteroclinic to a periodic orbit traveling
  wave solution to the KdV5 equation \cite{hoefer_modulation_2019}. }
\end{center}
\end{figure}

A recent preprint considers the so-called generalized Riemann problem
for the Whitham modulation equations of the Serre-Greene-Naghdi (SGN)
equations that model fully nonlinear shallow water waves
\cite{gavrilyuk_generalized_2018}. In particular, the numerical
evolution of initial conditions consisting of one periodic wave that
rapidly transitions to another periodic wave was shown as evidence
that discontinuous, weak solutions of the SGN-Whitham equations can be
realized. Emerging from this class of initial data were structures
that resemble genuine traveling wave solutions of the governing PDE.
The dispersion relation for the SGN equations is
$\omega(k) = k/\sqrt{k^2/3 +1}$, which does not exhibit an inflection
point. Consequently, the numerical results obtained in
\cite{gavrilyuk_generalized_2018} and their interpretation as shock
solutions of the Whitham equations suggest that Whitham shocks persist
in nonlocal, strongly nonlinear regimes in addition to the weakly
nonlinear cases studied here.

This paper is organized as follows. In Section
\ref{sec:kdv5-traveling-wave}, we discuss TW solutions of the KdV5
equation.  Homoclinic TW solutions are used to derive the KdV5-Whitham
modulation equations in Section \ref{sec:kdv5-whith-modul}.  We also
study the properties of the Whitham modulation equations and identify
a class of 2-wave curves that we use to solve the KdV5 equation
\eqref{eq:kdv5} with Riemann data \eqref{eq:riemann_data}. In Section
\ref{sec:WS_abstract}, we study shock solutions of the modulation
equations. Here, we also prove that heteroclinic TWs consisting of
disparate, co-propagating periodic waves necessarily satisfy the
Rankine-Hugoniot jump conditions for the Whitham equations.  We use
this result to define admissible Whitham shocks as the zero dispersion
limit of KdV5 heteroclinic TW solutions. In Section \ref{sec:WS}, we
use the Whitham shock locus to compute a rich variety of heteroclinic
and homoclinic TW solutions.  We discuss extensions and applications
of the newly developed theory to the Riemann problem
\eqref{eq:riemann_data} for the KdV5 equation \eqref{eq:kdv5}, the
Kawahara equation \eqref{eq:kawahara}, and the Whitham equation
\eqref{eq:whitham_eq_intro} in Sec.~\ref{sec:applications}. In Section
\ref{sec:conclusion}, we conclude the manuscript and postulate new and
related problems to the present work.

\section{KdV5 periodic traveling wave solutions}
\label{sec:kdv5-traveling-wave}

Traveling waves are sought in the form
\begin{align*}
  u(x,t) = f(\xi),\quad \xi = x - c t
\end{align*}
where $c$ is the phase velocity. The profile function $f$
satisfies the ODE
\begin{align}\label{eq:TW_ode}
- c f' + f f ' + f^{(5)} = 0. 
\end{align}
Integration yields the fourth order profile ODE
\begin{equation}
  \label{eq:11}
  f'''' + \frac{1}{2} f^2 - c f - A/2 = 0,
\end{equation}
where $A \in \mathbb{R}$ is a constant of integration.  This ODE can
be integrated again to yield the energy-type integral
\begin{equation}
  \label{eq:16}
  f''' f' - \frac{1}{2} \left ( f'' \right )^2 + \frac{1}{6} f^3 -
  \frac{c}{2} f^2 - \frac{A}{2} f + B = 0, 
\end{equation}
where $B \in \mathbb{R}$ is again a constant of integration.  We can
rewrite the ODE \eqref{eq:11} as the first order system
\begin{equation}
  \label{eq:12}
  \mathbf{\Phi}' =   L  \mathbf{\Phi} + R(\mathbf{\Phi}), \quad L =   \begin{pmatrix}
      0 & 1 & 0 & 0 \\
      0 & 0 & 1 & 0 \\
      0 & 0 & 0 & 1 \\
      c & 0 & 0 & 0
    \end{pmatrix}, \quad R(\mathbf{\Phi}) =   \begin{pmatrix}
      0 \\ 0 \\ 0 \\ \frac{1}{2}(A- \Phi_1^2)
    \end{pmatrix} ,
\end{equation}
where $\Phi_n = \mathrm{d}^{n-1} f/\mathrm{d} \xi^{n-1}$,
$n=1,2,3,4$, are the components of the vector $\mathbf{\Phi}$.

\subsection{Linearization}
\label{sec:linearization}

The fixed points $\mathbf{U}(\xi) = \mathbf{U}_0$ of
the ODE \eqref{eq:12} are
\begin{equation}
  \label{eq:13}
  \mathbf{U}_0^{(\pm)} =
  \begin{pmatrix}
    U_0^{(\pm)} & 0 & 0 & 0
  \end{pmatrix}^T , \quad U_0^{(\pm)} = c \pm \sqrt{c^2 + A},
\end{equation}
so that $c^2 > -A$ for real equilibria. Linearization about these
fixed points results in the matrices
\begin{equation}
  \label{eq:14}
  L + \frac{\partial \mathbf{R}}{\partial
    \mathbf{U}}(\mathbf{U}_0^{(\pm)}) =
  \begin{pmatrix}
    0 & 1 & 0 & 0 \\
    0 & 0 & 1 & 0 \\
    0 & 0 & 0 & 1 \\
    c - U_0^{(\pm)} & 0 & 0 & 0
  \end{pmatrix},
\end{equation}
whose eigenvalues $\lambda_j^{(\pm)}$ come in quartets with
corresponding right eigenvectors $\mathbf{r}_j^{(\pm)}$
\begin{equation}
  \label{eq:15}
  \begin{split}
    \lambda_1^{(\pm)} &= -\mu^{(\pm)}, \quad \mathbf{r}_1^{(\pm)} = \left(1,-\mu^{(\pm)},\mu^{(\pm)},-\left(\mu^{(\pm)}\right)^3 \right)^{\rm T}, \\
    \lambda_2^{(\pm)} &= -i\mu^{(\pm)}, \quad
    \mathbf{r}_2^{(\pm)} = \left(1,-i\mu^{(\pm)},-\mu^{(\pm)},i\left(\mu^{(\pm)}\right)^3 \right)^{\rm T} ,\\
    \lambda_3^{(\pm)} &= i\mu^{(\pm)}, \quad \mathbf{r}_3^{(\pm)} =  \left(1,\mu^{(\pm)},-\mu^{(\pm)},-i\left(\mu^{(\pm)}\right)^3 \right)^{\rm T} ,\\
    \lambda_4^{(\pm)} &= \mu^{(\pm)}, \quad \mathbf{r}_4^{(\pm)} =
    \left(1,\mu^{(\pm)},\mu^{(\pm)},\left(\mu^{(\pm)}\right)^3 \right)^{\rm
      T},
  \end{split}
\end{equation}
where $\mu^{(\pm)} = (c - U_0^{(\pm)})^{1/4}$.  We observe that the fixed point $\mathbf{U}_0^{(-)}$ has
$U_0^{(-)} = c - \sqrt{c^2 +A} < c$ so that the complex conjugate
eigenvalues $\lambda_{2,3}^{(-)} = \mp i (c^2+A)^{1/8}$ lie on the
imaginary axis whereas $\lambda_{1,4}^{(-)} = \mp (c^2 +A)^{1/8}$ are
opposite and lie on the real axis. The eigenvalues $\lambda_j^{(+)}$
are $\pi/4$ rotations of $\lambda_j^{(-)}$ in the complex plane for
$j = 1,2,3,4$.

This linear analysis suggests the existence of small amplitude
periodic orbits near $\mathbf{U}_0^{(-)}$, with one-dimensional stable
and unstable manifolds in the directions $\mathbf{r}_1^{(-)}$ and
$\mathbf{r}_4^{(-)}$, respectively. These manifolds are what enable
the periodic-heteroclinic-to-periodic TWs numerically computed in this
paper.  The existence of periodic orbits for \eqref{eq:12} was
demonstrated in \cite{vanderbauwhede_homoclinic_1992} and related
results establishing the existence and stability of wavetrains in the
related Kawahara equation \eqref{eq:kawahara} are established in
\cite{haragus_spectral_2006}. This motivates our investigation of
heteroclinic connections between differing periodic orbits.

First, we consider $2\pi$-periodic traveling wave solutions to
\eqref{eq:kdv5} in the form
\begin{equation}
  \label{eq:18}
  f(\xi) = \varphi(\theta; \ub,a,k), \quad
  \varphi(\theta + 2\pi; \ub,a,k) = \varphi(\theta; \ub,a,k),
\end{equation}
where $\theta = k\xi$ is the phase variable and $\varphi$ possesses
three distinct parameters identified in Figure \ref{fig:qual_per} and
defined as
\begin{align*}
  \text{wavenumber:}& \quad k,\\
  \text{ wave mean:}& \quad \ub = \frac{1}{2\pi}\int_0^{2\pi}
                                 \varphi(\theta) d \theta,\\ 
  \text{wave amplitude:}& \quad a = \max\limits_{\theta \in [0,\pi]} \varphi(\theta) - \min\limits_{\theta \in [0,\pi]}  \varphi(\theta),
\end{align*}
and wave frequency $\omega = ck$, $c = c(\ub,a,k)$.

\begin{figure}
\begin{center}
\includegraphics[scale=0.4]{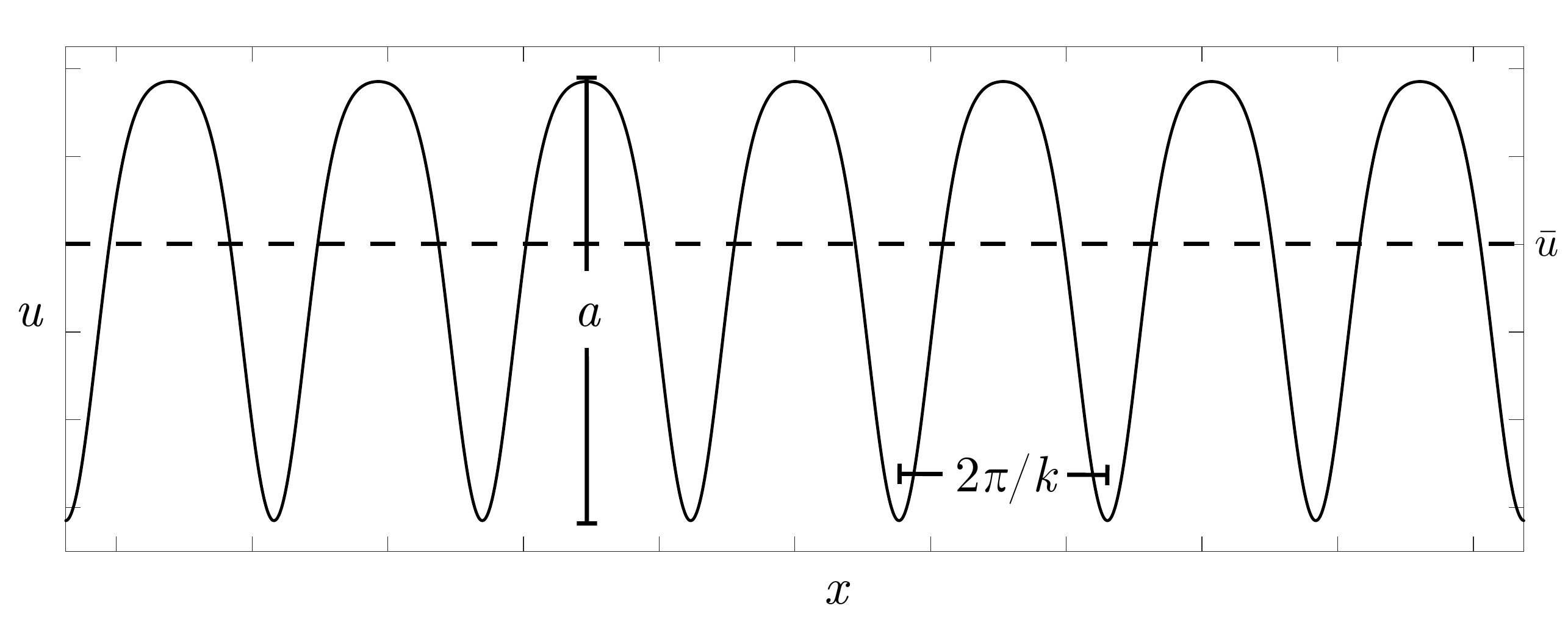}
\caption{Nonlinear periodic traveling wave solution to \eqref{eq:kdv5} with mean $\ub$, amplitude $a$ and wavenumber $k$.}
\label{fig:qual_per}
\end{center}
\end{figure}

\subsection{Scaling and Galilean symmetries}
\label{sec:scal-galil-symm}

KdV5 \eqref{eq:kdv5} admits the Galilean
\begin{equation}
  \label{eq:19}
  u(x,t) \to u(x-u_0t,t) + u_0,
\end{equation}
and scaling
\begin{equation}
  \label{eq:20}
  u \to b u, \quad x \to b^{-1/4} x, \quad t \to b^{-5/4} t
\end{equation}
symmetries. Therefore, the three-parameter family of periodic traveling waves can be generated from the one-parameter family
$\left \{ (\tilde{\varphi}(\theta;\tilde{k}),\tilde{c}(\tilde{k})) \right
\}_{\tilde{k}\ge 0}$ by
\begin{equation}
  \label{eq:21}
  \begin{split}
    \varphi(\theta;\ub,a,k) &= \ub +
    a\tilde{\varphi}(a^{1/4}\theta-a^{1/4}k \ub
    t;a^{-1/4}k)  ,\\
    c(\ub,a,k) &= a \tilde{c}(a^{-1/4}k) + \ub ,
  \end{split}
\end{equation}
for any $a > 0$, $k > 0$, $\ub \in \mathbb{R}$. Note that $\tilde{\varphi}(\theta;\tilde{k}) = \varphi(\theta;0,1,\tilde{k})$ and $\tilde{c}(\tilde{k})  = c (0,1,\tilde{k})$.

\subsection{Approximate and numerical computations of periodic waves}
\label{sec:stok-wave-appr}
We now obtain the periodic solutions to \eqref{eq:TW_ode} via a weakly
nonlinear Stokes frequency shift calculation
\cite{whitham_linear_1974} and through numerical computations. The
periodic solutions will be used to obtain the Whitham modulation
system.

To this end, we seek an approximate form for the periodic wave
$\varphi(\theta)$ and its phase speed $c$ as series expansions in the
small, finite amplitude parameter $a$
\begin{align*}
  \varphi = \ub + a \varphi_1 + a^2 \varphi_2 + \ldots, \quad c = c_0 + a^2 c_2 + \ldots.
\end{align*}
Inserting the asymptotic approximation and equating like powers of
$a$, we obtain the periodic wave and phase velocity up to
$\mathcal{O}(a^2)$
\begin{align}
  \label{eq:stokes}
  \varphi(\theta; \ub,a,k) &= \ub + \frac{a}{2} \cos \theta -
  \frac{a^2}{240k^4}\cos 2 \theta + \ldots, \quad 0 < a \ll 1, \quad k \gg
  a^{1/4}, \\
  \label{eq:stokes_disp}
  c(\ub,a,k) &= \ub +  k^4 - \frac{a^2}{480 k^4} + \ldots, \quad 0 <
  a \ll 1, \quad  k \gg a^{1/4}.
\end{align}
Here, the restriction on the wavenumber $k \gg a^{1/4}$ is required so
that the asymptotic series remains well ordered. We remark that the
asymptotic series can be rescaled to arbitrary $a$ via the symmetries
\eqref{eq:19}, \eqref{eq:20} so that the requirement $a \ll 1$ can be
formally relaxed but we must maintain the short wave requirement
$k \gg a^{1/4}$ in order to respect asymptotic ordering.

We now present computations of periodic solutions to the profile ODE
\eqref{eq:TW_ode} where we take $A = 0$ via a Galilean shift.  The
solutions are computed via a pesudospectral projection onto a Fourier
basis. A solution to the nonlinear system for the periodic wave's
Fourier coefficients is carried out with a Newton iteration. Details
of this computation are contained in Appendix A.

Examples of periodic traveling wave solutions and their corresponding
nonlinear dispersion curve are given in
Fig.~\ref{fig:per_waves}. Small $\tilde{k}$ corresponds to a well
separated, solitary wave train.  One metric for the validity of the
Stokes approximation is the accuracy of the phase velocity
\eqref{eq:stokes_disp}, which compares favorably with the numerically
computed phase velocity for $\tilde{k} \gtrsim 0.4$.

\begin{figure}
\begin{center}
\includegraphics[scale=0.4]{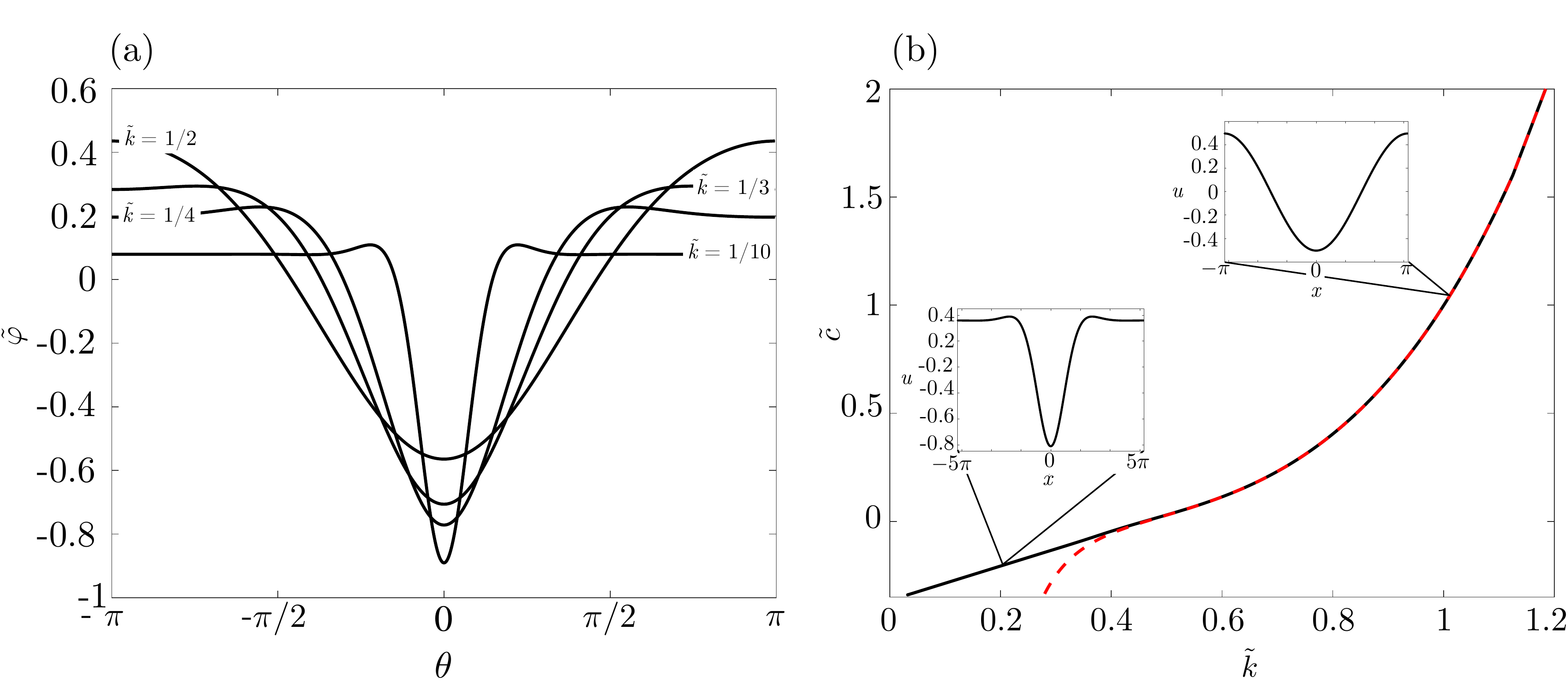}\label{fig:per_waves}
\caption{Numerically computed periodic traveling waves
  $\tilde{\varphi}(\theta,\tilde{k})$ with zero mean and unit
  amplitude. (a) Example periodic waves. (b) Comparison of numerically
  computed phase velocity $\tilde{c}$ for the family of periodic waves
  (black solid curve) compared with the Stokes approximation
  Eq.~\eqref{eq:stokes_disp} (red dashed curve). Insets: example
  periodic waves for two different phase velocities.}
\end{center}
\end{figure}

\section{KdV5-Whitham modulation equations}
\label{sec:kdv5-whith-modul}

The Whitham modulation equations describe the dispersionless limit of nonlinear wavetrains. Equivalently, these equations describe slow modulations of periodic traveling waves. The modulation equations are a system of first order, quasi-linear conservation laws for the parameters $(\ub,a,k)$ that we now derive. Equation \eqref{eq:kdv5} admits the two conservation laws
\begin{align}
  \label{eq:1}
  (u)_t + \left( \frac{1}{2}u^2 + u_{xxxx}\right)_x & = 0, \\
  \label{eq:2}
  \left(\frac{1}{2} u^2\right)_t + \left(\frac{1}{3}u^3 + uu_{xxxx} -
    u_x u_{xxx}  + \frac{1}{2}u_{xx}^2\right)_x & = 0.
\end{align} 
The Whitham equations can be obtained by a period-average of the
conservation laws, evaluated on the periodic TW manifold
\cite{whitham_non-linear_1965}. To see this, we introduce the rapid
phase variable $\theta = S(X,T)/\epsilon$ satisfying
\begin{align*}
  k = \theta_x = S_X, \quad \omega = -\theta_t = -S_T
\end{align*}
where we recall $X = \epsilon x$ and $T = \epsilon t$ are the long
space and slow time scales with $ 0 < \epsilon \ll 1$ as in
\eqref{eq:kdv5_dispersionless}. We then seek an asymptotic solution in
the form
\begin{align}\label{eq:asymp_series}
u(x,t) = \varphi(\theta; \ub,a,k) + \epsilon u_1(\theta,X,T) + \mathcal{O}(\epsilon^2),
\end{align} 
with a leading order periodic traveling wave whose parameters $\ub$,
$a$, $k$ depend on the slow variables $X$ and $T$ and each of the
$u_n$ are $2\pi$-periodic in $\theta$. Consequently, the
derivatives in \eqref{eq:1}, \eqref{eq:2} transform according to
\begin{align*}
\frac{\partial}{\partial x} = k\frac{\partial}{\partial \theta} + \epsilon \frac{\partial}{\partial X}, \quad \frac{\partial}{\partial t} = -\omega\frac{\partial}{\partial \theta} + \epsilon \frac{\partial}{\partial T}. 
\end{align*}
We now introduce averaging of the quantity $F[u(x,t);X,T]$ according
to
\begin{align}\label{eq:avg_def}
  \overline{F} = \frac{1}{2\pi}\int_{0}^{2\pi} F  \ d\theta .
\end{align}
Then, substituting the asymptotic series \eqref{eq:asymp_series} for
$u$ into $F$ results in
\begin{align*}
  \overline{\frac{\partial F}{\partial t}} = \epsilon \frac{\partial \overline{F}}{\partial T} + \mathcal{O}(\epsilon^2), \quad
  \overline{\frac{\partial F}{\partial x}} = \epsilon \frac{\partial \overline{F}}{\partial X} + \mathcal{O}(\epsilon^2),
\end{align*}
owing to the fact that the period average of an exact $\theta$
derivative is zero for periodic functions.  Applying the averaging
operator \eqref{eq:avg_def} to the conservation laws
\eqref{eq:1}, \eqref{eq:2} with the asymptotic series
\eqref{eq:asymp_series} for $u$, we obtain
\begin{align}
  \label{eq:5}
  (\overline{u})_T + \left( \frac{1}{2}\ol{\varphi^2} +
  k^4\ol{\varphi_{\theta\theta\theta\theta}}\right)_X & = 0, \\
  \label{eq:6}
  \left(\frac{1}{2} \ol{\varphi^2}\right)_T + \left(\frac{1}{3}\ol{\varphi^3} +
  k^4 \ol{\varphi \varphi_{\theta\theta\theta\theta}} - k^4\ol{\varphi_\theta
  \varphi_{\theta\theta\theta}} +  \frac{1}{2}k^4
  \ol{\varphi_{\theta\theta}^2}\right)_X & = 0,  
\end{align}
subject to $\mathcal{O}(\epsilon)$ corrections that vanish in the zero
dispersion limit.  The modulation system is closed by the phase
compatibility condition $S_{XT} = S_{TX}$, yielding the conservation
of waves
\begin{align}
  \label{eq:3}
  k_T + (c k)_X = 0 .
\end{align}
Thus, we conclude that the KdV5-Whitham modulation equations
\eqref{eq:5}, \eqref{eq:6}, and \eqref{eq:3} describe the
dispersionless limit of the KdV5 equation
\eqref{eq:kdv5_dispersionless}.  If initial data for
Eqs.~\eqref{eq:5}, \eqref{eq:6}, and \eqref{eq:3} are invariant to the
hydrodynamic scaling transformation $X \to b X$, $T \to b T$---e.g.,
the step initial data considered later---then an equivalent
alternative to the zero dispersion limit is the long time limit of
Eqs.~\eqref{eq:5}, \eqref{eq:6}, and \eqref{eq:3} with the unscaled
independent variables $x = X/\epsilon$, $t = T/\epsilon$.

The conservation laws \eqref{eq:5}, \eqref{eq:6} can be simplified via
integration by parts
\begin{align*}
  \overline{\varphi_{\theta\theta\theta\theta}} &= 0, \qquad
  \overline{\varphi\varphi_{\theta\theta\theta\theta}} =
                                                            \overline{\varphi_{\theta\theta}^2}, \qquad
  \overline{\varphi_{\theta} \varphi_{\theta\theta\theta}} =
                                                             -\overline{\varphi_{\theta\theta}^2} .
\end{align*}
These identities, along with a return to the unscaled variables $x,t$, imply that the Whitham equations in conservative form can be written compactly
\begin{align}
  \label{eq:whitham_1}
  \overline{u}_t + \left( \frac{1}{2}\ol{\varphi^2} \right)_x & = 0, \\
  \label{eq:whitham_2}
  \left(\frac{1}{2} \ol{\varphi^2}\right)_t + \left(\frac{1}{3}\ol{\varphi^3}  +
    \frac{5}{2}k^4 \ol{\varphi_{\theta\theta}^2}\right)_x & = 0,  \\
  \label{eq:whitham_3}
  k_t+ (ck)_x & = 0 .
\end{align}
This will be convenient of our purposes because it eliminates the
small dispersion parameter $\epsilon$ from the problem.

\subsection{Properties of the Whitham equations}
\label{sec:prop-whith-equat}

\subsubsection{Weakly nonlinear regime}
\label{sec:weakly-nonl-regime}

We approximate the averaging integrals in the Whitham equations
\eqref{eq:whitham_1}--\eqref{eq:whitham_3} with the Stokes wave
\eqref{eq:stokes} and its dispersion relation \eqref{eq:stokes_disp}
yielding
\begin{align}
  \label{eq:27}
  \ol{\varphi^2} &= \ub^2 + \frac{a^2}{8} + \cdots, \quad  \ol{\varphi^3} = \ub^3 + \frac{3}{8} \ub a^2 + \cdots ,\quad 
  \ol{(\varphi'')^2} = \frac{a^2}{8} + \cdots ,
\end{align}
accurate to $\mathcal{O}(a^2)$. Inserting these approximate integrals
into the averaged conservation laws
\eqref{eq:whitham_1}--\eqref{eq:whitham_3}, the weakly nonlinear
KdV5-Whitham equations in conservative form are approximately
\begin{align}\label{eq:weaklyNLconserv}
\begin{split}
\ub_t + \left(\frac{1}{2} \ub^2 + \frac{a^2}{16} \right)_x & = 0,\\
\left(\frac{1}{2} \ub^2 + \frac{a^2}{16} \right)_t + \left(\frac{1}{3}\ub^3 + \frac{1}{8}\ub a^2 + \frac{5}{16}k^4a^2\right)_x & = 0, \\
k_t + \left[\left(\ub + k^4 - \frac{a^2}{480k^4}\right)k\right]_x & = 0. 
\end{split}
\end{align}
The non-conservative, quasi-linear form is
\begin{align}
  \label{eq:25}
  \begin{bmatrix}
    \ub \\ a \\ k 
  \end{bmatrix}_t + \begin{bmatrix}
    \ub & \frac{a}{8} & 0 \\
   a  & \ub + 5k^4 & 10 a k^3 \\
    k     & 		-\frac{a}{240 k^3} & \ub + 5k^4 + \frac{a^2}{160k^4} \\
  \end{bmatrix}\begin{bmatrix}
    \ub \\ a \\ k 
  \end{bmatrix}_x = 0 .
\end{align}
A standard perturbation calculation gives the approximate eigenvalues
\begin{align}
  \label{eq:30}
  \lambda_1 & = \ub + \frac{a^2}{40k^4} + \mathcal{O}(a^3), \\
  \label{eq:31}
  \lambda_2 & = \ub + 5k^4 - \sqrt{\frac{5}{24}}a -\frac{3a^2}{320k^4}+ \mathcal{O}(a^3),
  \\
  \label{eq:32}
  \lambda_3 & = \ub + 5k^4 +  \sqrt{\frac{5}{24}}a -\frac{3a^2}{320k^4}+ \mathcal{O}(a^3),
\end{align}
and eigenvectors
\begin{align}
  \mathbf{r}_1  &= \begin{pmatrix}
    -5k^3 \\ 0 \\1 
  \end{pmatrix}  - a \begin{pmatrix}
    0 \\ \tfrac{1}{k} \\0 
  \end{pmatrix} +  a^2 \begin{pmatrix}  \tfrac{7}{480k^5} \\ 0\\ 0
  \end{pmatrix}+ \mathcal{O}(a^3), \quad\\
  \mathbf{r}_2 & = \begin{pmatrix}0 \\ -4\sqrt{30}	k^3 \\ 1
  \end{pmatrix} 
  + a \begin{pmatrix}-\frac{\sqrt{3}}{\sqrt{10}k}\\ \frac{33}{20k} \\ 0\end{pmatrix}
  + a^2
  \begin{pmatrix}
  -\frac{7}{800 k^5}\\ \sqrt{\frac{3}{10}}\frac{213 }{3200 k^5} \\0 
  \end{pmatrix},+ \mathcal{O}(a^3)\\
  \mathbf{r}_3 & = \begin{pmatrix} 0 \\  4\sqrt{30}	k^3 \\ 1
  \end{pmatrix} 
   + a \begin{pmatrix} \frac{\sqrt{3}}{\sqrt{10}k}\\ \frac{33}{20k} \\ 0\end{pmatrix}
   +  a^2 \begin{pmatrix}
  -\frac{7}{800 k^5}\\ -\sqrt{\frac{3}{10}}\frac{213 }{3200 k^5}\\0 
  \end{pmatrix}+ \mathcal{O}(a^3)
\end{align}
of the flux Jacobian matrix in \eqref{eq:25}. The quasi-linear system
\eqref{eq:25} is genuinely nonlinear if
$\mu_j = \nabla \lambda_j \cdot \mathbf{v}_j \neq 0$ for all $j$. By
inspection and direct calculation, we find that the weakly nonlinear
Whitham equations in the asymptotic regime $0 < a \ll 1$,
$k \gg a^{1/4}$ are strictly hyperbolic and genuinely nonlinear.


\subsubsection{Strongly nonlinear regime}
\label{sec:strongly-nonl-regime}
In the strongly nonlinear regime, we rewrite the modulation equations \eqref{eq:whitham_1}--\eqref{eq:whitham_3} in non-conservative form in terms of the modulation variables $\mathbf{q} = [\ub, a, k]^{\rm T}$ 
\begin{align}\label{eq:whitham_noncons}
\mathbf{q}_t + \mathcal{A}\mathbf{q}_x = 0 
\end{align}
where 
\begin{align*}\displaystyle
\mathcal{A} = 
\begin{bmatrix}
1 & 0 & 0 \\
\frac{1}{2}\frac{\partial \ol{\varphi^2} }{\partial \ub} & \frac{1}{2}\frac{\partial \ol{\varphi^2} }{\partial a} & \frac{1}{2}\frac{\partial \ol{\varphi^2} }{\partial k}\\
0 & 0 & 1
\end{bmatrix}^{-1}\begin{bmatrix}
\frac{1}{2}\frac{\partial \ol{\varphi^2} }{\partial \ub} & \frac{1}{2}\frac{\partial \ol{\varphi^2} }{\partial a} & \frac{1}{2}\frac{\partial \ol{\varphi^2} }{\partial k}\\
\frac{1}{3}\frac{\partial \ol{\varphi^3} }{\partial \ub} + \frac{5}{2}k^4\frac{\partial \ol{\varphi_{\theta\theta}^2} }{\partial \ub} & \frac{1}{3}\frac{\partial \ol{\varphi^3} }{\partial a} + \frac{5}{2}k^4\frac{\partial \ol{\varphi_{\theta\theta}^2} }{\partial a} & \frac{1}{3}\frac{\partial \ol{\varphi^3} }{\partial k}+ \frac{5}{2}\frac{\partial (k^4 \ol{\varphi_{\theta\theta}^2} )}{\partial \ub}\\
k \frac{\partial c}{\partial \ub} & k \frac{\partial c}{\partial a} & \frac{\partial( c k)}{\partial \ub}
\end{bmatrix},
\end{align*}
which is valid provided the matrix inverse exists. For all of the
computations performed here, $\mathcal{A}$ is well-defined.  The
symmetries \eqref{eq:19}, \eqref{eq:20} of solutions to the KdV5
equation can be used to directly compute the integral dependence on
the mean $\ub$ and amplitude $a$.  We define the averaging integrals
\begin{align*}
  I_n(\ub,a,k) = \frac{1}{2\pi}\int_0^{2\pi} \varphi^n(\theta;\ub,a,k)  \ d \theta, \quad J_2(\ub,a,k) =  \frac{1}{2\pi}\int_0^{2\pi} \varphi_{\theta\theta}^2(\theta;\ub,a,k) \  d \theta
\end{align*}
These integrals can be written explicitly in terms of period averages
over the one-parameter family of periodic solution
$\tilde{\varphi}(\theta;\tilde{k})$ using \eqref{eq:21}
\begin{equation}\label{eq:scaled_ints}
  \begin{split}
    I_2(\ub,a,k) &= \ub^2 + a^2 \tilde{I}_2(a^{-1/4}k), \\
    I_3(\ub,a,k) &= \ub^3 + 3\ub a^2 \tilde{I}_2(a^{-1/4}k) + a^3
    \tilde{I}_3(a^{-1/4}k), \\
    J_2(\ub,a,k) &= a^2\tilde{J}_2(a^{-1/4}k),
  \end{split}
\end{equation}
where we have introduced the $\tilde{\cdot}$ variables to denote
averaging integrals of unit amplitude, zero mean periodic
solutions.

We now investigate the hyperbolicity/ellipticity and genuine
nonlinearity of the modulation equations. These criteria depend solely
upon the eigenvalues $\lambda_j$ and eigenvectors $\mathbf{r}_j$ of
$\mathcal{A}$ that satisfy
\begin{align*}
\left(\mathcal{A}-\lambda_j I\right) \rb_j = 0,
\end{align*}
where $I$ is the $3 \times 3$ identity matrix. In general, we expect
three eigenpairs $\left\{ (\lambda_j,\rb_j)\right\}_{j = 1}^{3}$ that
depend on the values $(\ub,a,k)$. By the symmetries \eqref{eq:19},
\eqref{eq:20}, it is enough to determine their dependence on
$\tilde{k}$ only.

Weak hyperbolicity of the modulation equations is a necessary
condition for the modulational stability of periodic waves
\cite{whitham_linear_1974,benzoni-gavage_slow_2014}. Numerical
computation of the characteristic velocities depicted in
Figure~\ref{fig:characteristics} illustrate hyperbolicity of the
modulation equations for a range of $\tilde{k}$. When
$0.41 \lesssim \tilde{k} \lesssim 0.65$ and
$0.25 \lesssim \tilde{k} \lesssim 0.3$, two complex conjugate
characteristic velocities with nonzero imaginary part emerge, so that
corresponding periodic waves are modulationally unstable. It is
further shown in Fig.~\ref{fig:characteristics} that the weakly
nonlinear calculations for the characteristic velocities
Eq.~\eqref{eq:30}--\eqref{eq:32} are in excellent agreement with fully
nonlinear calculations for $\tilde{k} \gtrsim 0.7$.

In a similar manner, we utilize the eigenvalues and eigenvectors to determine regions of genuine nonlinearity for the modulation equations \eqref{eq:whitham_1}--\eqref{eq:whitham_3} by computing the quantity 
\begin{align}\label{eq:genNL}
  \mu_j =\nabla \lambda_j \cdot \mathbf{r}_j, \quad j = 1,2,3.
\end{align}
For waves of unit amplitude and zero mean, we find points of
non-genuine nonlinearity where $\mu_{1} \approx \mu_{2} \approx 0$,
which are identified by the red points in the lower panel of
Fig.~\ref{fig:characteristics} (a).  Our careful numerical
computations (see Appendix) suggest that both
the first and second characteristic fields are not genuinely nonlinear
at these select points $\tilde{k} \in \{0.167,0.185,0.222,0.370\}$. 

In strictly hyperbolic regions, we compute wave curves that
parameterize self-similar, simple wave solutions to the modulation
equations \eqref{eq:whitham_noncons}. These wave curves will be
applied to the solution of Riemann problems in
Sec.~\ref{sec:applications}. The $j$-th wave curve is the integral
curve in the direction of $\mathbf{r}_j$
\cite{dafermos_hyperbolic_2009}. We introduce the self similar
parameterization $s = s(x,t)$ so that $\mathbf{q} = \mathbf{q}(s)$ and,
along $s = \lambda_j$,
\begin{align}\label{eq:simp_wave_ode}
  \frac{\mathrm{d}\mathbf{q}}{\mathrm{d}s} = \frac{\mathbf{r}_j}{\mu_j}.
\end{align}
Illustrative 2-wave curves are displayed in
Figs.~\ref{fig:characteristics}(b) and (c) and correspond to the
shaded regions in Fig.~\ref{fig:characteristics}(a) (lower panel) on
the left and right respectively. In these figures, we project the
three dimensional curve $(\ub(s),a(s),k(s))$ onto both the $\ub$-$a$
plane and the $\ub$-$k$ plane. We normalize the wave curves to emanate
from the zero mean, unit amplitude state $\ub = 0$, $a = 1$. For this
visualization, we should note that $\ub(s)$ is a monotonically
decreasing function in $s$.

\begin{figure}[h!]
\begin{center}
\includegraphics[scale = 0.5]{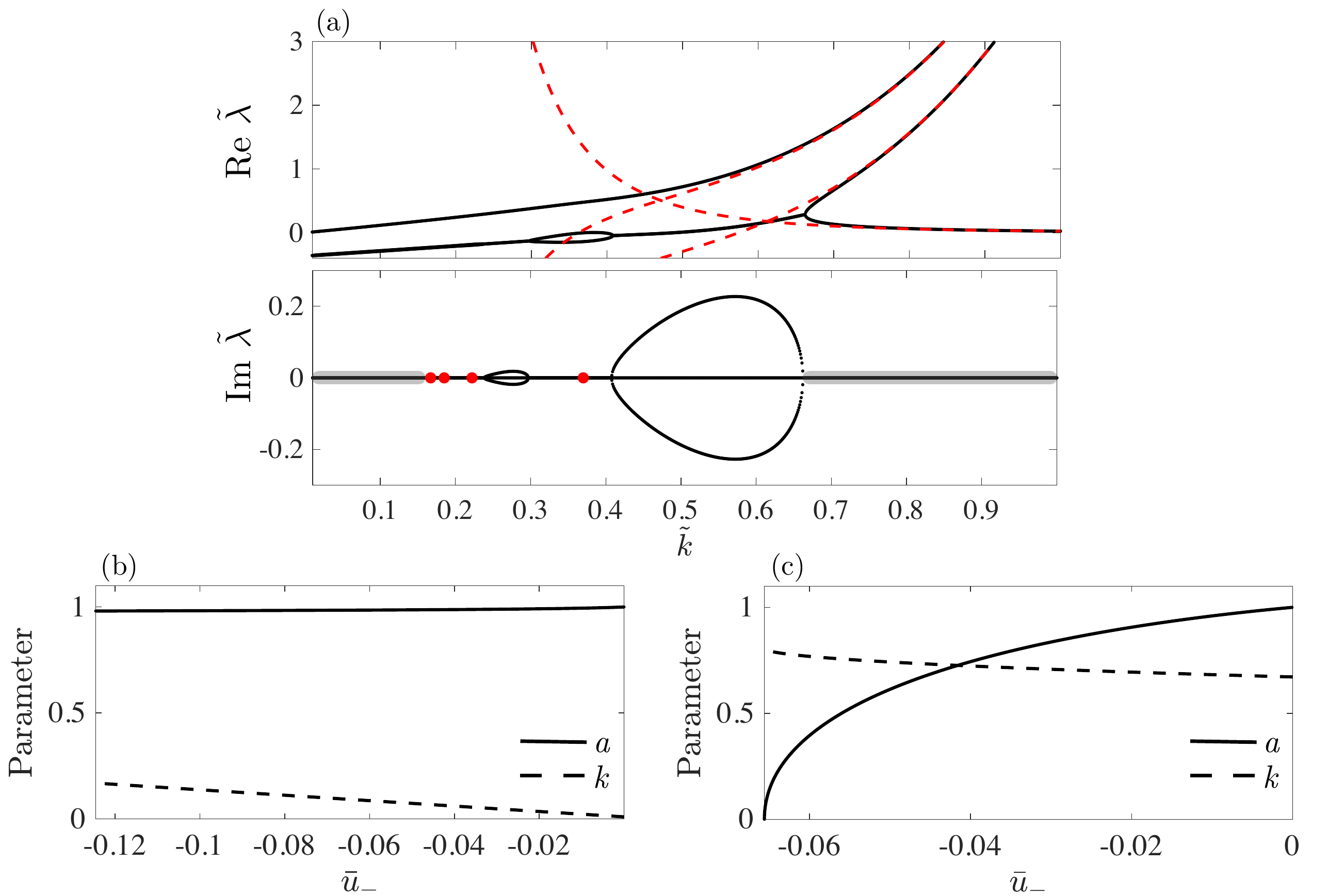}
\caption{(a) Eigenvalues of $\mathcal{A}$ for $\ub = 0$, and $a = 1$
  as a function of wavenumber $\tilde{k}$ shown in black curves. The
  eigenvalues from weakly nonlinear theory
  \eqref{eq:30}--\eqref{eq:32} are identified by dashed red curves.
  Gray, banded regions denote the domain of the 2-wave curves in (b)
  and (c). (b) Example 2-wave curve for $0 < \tilde{k} \lesssim 0.15$.
  (c) Example 2-wave curve for $0.65 \lesssim \tilde{k} < 1$.}
\label{fig:characteristics}
\end{center}
\end{figure}

\section{Whitham shocks: Abstract Setup}
\label{sec:WS_abstract}
Now that we have obtained the Whitham modulation equations and
described their structure, we consider shock solutions for the
conservation laws \eqref{eq:whitham_1}--\eqref{eq:whitham_3} and their
connections to heteroclinic traveling wave solutions of the KdV5
equation.

Shock solutions of the modulation equations
\eqref{eq:whitham_1}--\eqref{eq:whitham_3} take the form of a moving
discontinuity
\begin{align}
  \label{eq:BCs}
  \ol{u}(x,t) & =
  \begin{cases}
    \ol{u}_- & x < Vt \\
    \ol{u}_+ & x > Vt
  \end{cases}
  &a(x,t) & =
  \begin{cases}
    a_- & x < Vt \\
    a_+ & x > Vt
  \end{cases}
  &k(x,t)  =
  \begin{cases}  
    k_- & x < Vt \\
    k_+ & x > Vt
  \end{cases}.
\end{align}
In order for the shock \eqref{eq:BCs} to be a weak solution of
\eqref{eq:whitham_1}--\eqref{eq:whitham_3}, the left
$(\ub_-,a_-,k_-)$, right $(\ub_+,a_+,k_+)$ states and the shock
velocity $V$ satisfy the Rankine-Hugoniot jump conditions
$$ - V \llbracket \mathbf{P} \rrbracket  + \llbracket  \mathbf{Q} \rrbracket = 0, $$ 
where $\mathbf{P}, \mathbf{Q} \in \mathbb{R}^3$ are the averaged
densities and fluxes in
Eqs.~\eqref{eq:whitham_1}--\eqref{eq:whitham_3} and
$\llbracket \cdot \rrbracket $ represents the difference in the left,
right values. Consequently, the jump conditions take the form
\begin{align}
  -V(\ol{u}_- - \ol{u}_+ ) + \frac{1}{2}\left(\ol{\varphi_-^2} -
  \ol{\varphi_+^2} \right)  & = 0,\label{eq:RH1}\\ 
  -\frac{V}{2}\left(\ol{\varphi_-^2} - \ol{\varphi_+^2} \right) +
  \frac{1}{3}\left(\ol{\varphi_-^3} - \ol{\varphi_+^3} \right) +
  \frac{5}{2}\left(k_-^4  \ol{\varphi_{-,\theta\theta}^2} -
  k_+^4\ol{\varphi_{+,\theta\theta}^2} \right) & = 0,\label{eq:RH2}\\ 
  -V(k_- - k_+) + \left(c_- k_-  - c_+ k_+\right) & = 0 . \label{eq:RH3}
\end{align}
where $\varphi_\pm = \varphi(\theta; \ub_\pm, a_\pm,k_\pm)$ are the
right ($+$)/left ($-$) periodic traveling waves that compose the
Whitham shock.  We summarize this in the following.
\begin{mydef}[Whitham shock]
  \label{sec:whith-shocks:-abstr}
  A discontinuity \eqref{eq:BCs} moving with velocity $V$ is called a
  \textit{Whitham shock} when it is a weak solution of the Whitham
  modulation equations in conservative form
  \eqref{eq:whitham_1}--\eqref{eq:whitham_3}.  That is, the left
  $(\ub_-,a_-,k_-)$ and right $(\ub_+,a_+,k_+)$ states satisfy the
  Rankine-Hugoniot jump conditions \eqref{eq:RH1}--\eqref{eq:RH3}.
\end{mydef}

In order to determine admissibility of Whitham shocks, we consider
smooth heteroclinic TW solutions to the KdV5 equation \eqref{eq:kdv5}
that asymptote to distinct periodic waves as $x \to \pm \infty$
\begin{align}\label{eq:IC}
  u(x,t) \to \begin{cases} \varphi_-(k_-x  - \omega_- t -\theta_-) \equiv
    \varphi(k_-x  - \omega_- t - \theta_-;\ub_-,a_-,k_-)
    & x \to -\infty \\[3mm] \varphi_+(k_+x-\omega_+ t - \theta_+) \equiv \varphi(k_+x - \omega_+ t - \theta_+;\ub_+,a_+,k_+) & x \to \infty  \end{cases} ,
\end{align}
with each far-field periodic wave ($\varphi_\pm$) characterized by the
parameters $(\ub_\pm,a_\pm,k_\pm)$ and phases $\theta_\pm$.

Our main result states that KdV5 heteroclinic traveling wave solutions
exhibiting the far field behavior \eqref{eq:IC} necessarily lie on the
Whitham shock locus of states satisfying
\eqref{eq:RH1}--\eqref{eq:RH3}.
\begin{theorem}
  \label{thm:1}%
  Suppose $f(\xi)$ with speed $c$ ($\xi = x -ct$) satisfies
  \eqref{eq:TW_ode}, hence is a traveling wave solution of the
  differential equation \eqref{eq:kdv5} such that
  $$\inf\limits_{\theta_\pm \in [0, 2\pi)} |f(\xi) -
  \varphi_\pm(k_\pm\xi - \theta_\pm) | \to 0$$ uniformly as
  $\xi \to \pm \infty$, then the Rankine-Hugoniot relations
  \eqref{eq:RH1}--\eqref{eq:RH3} for the KdV5-Whitham equations
  \eqref{eq:whitham_1}--\eqref{eq:whitham_3} are satisfied by the
  far-field periodic waves $\varphi_\pm$ with parameters $\ub_\pm$,
  $a_\pm$, $k_\pm$.  The traveling wave speed is the shock speed
  $c = V$ and coincides with the phase velocities $c = c_\pm$.
\end{theorem}
\begin{proof}
  By assumption, there exists some TW profile $f(\xi)$ that
  solves equation \eqref{eq:TW_ode} with the far-field behavior
  described by the boundary conditions \eqref{eq:IC}. The factors
  $\theta_\pm$ are appropriate phase shifts so that the traveling wave
  matches the far-field periodic orbits $\varphi_\pm$. We define the
  averaging operators $\bar{\mathcal{L}}^{\pm}$ acting on a periodic,
  integrable function $F(\xi)$ as
  \begin{align*}
    \bar{\mathcal{L}}^{\pm} [F] & = \lim\limits_{\bar{x} \to \pm \infty}
                            \frac{k_\pm}{2\pi} \int_{\bar{x}}^{\bar{x} +
                            2\pi/k_\pm} F(\xi) d\xi  .
  \end{align*}
  Since $f \to \varphi_\pm$ uniformly, we can compute for $n = 1,2,3$
 
  \begin{align*}
    \bar{ \mathcal{L}}^{\pm} [f^n] & =  \lim\limits_{\bar{x} \to \pm \infty}
                                           \frac{k_\pm}{2\pi}\int_{\bar{x}}^{\bar{x}
                                           + 2\pi/k_\pm} f^n(\xi) d \xi, \\ 
                                         &  = \frac{1}{2\pi}\int_{\bar{x}}^{\bar{x}
                                           + 2\pi} \varphi_\pm^n(\theta) d \theta, \\
                                         &= \ol{\varphi_\pm^n}
  \end{align*}
  and
  \begin{align*}
    \bar{\mathcal{L}}^{\pm} [(f'')^2] & =  \lim\limits_{\bar{x} \to \pm \infty}
                                              \frac{k_\pm}{2\pi}\int_{\bar{x}}^{\bar{x}
                                              + 2\pi/k_\pm} (f''(\xi))^2 d \xi \\ 
                                            &  = \frac{k_\pm^4}{2\pi}\int_{\bar{x}}^{\bar{x}
                                              + 2\pi}  \left(\varphi_\pm''(\theta)
                                              \right)^2 d \theta \\
                                            &= k_\pm^4 \ol{(\varphi_\pm'')^2} .
  \end{align*}
  Since $f$ is a traveling wave solution to \eqref{eq:TW_ode}
  with two periodic wave limits, the phase speed of each must be the
  same, i.e.,
  \begin{equation}
    \label{eq:23}
    c = c_\pm .
  \end{equation}
  This condition immediately implies the jump condition \eqref{eq:RH3}
  with $V = c$ from the conservation of waves. All TW solutions with
  speed $c$ admit the first and second integrals \eqref{eq:11} and
  \eqref{eq:16}.
  We now apply the operator $\bar{\mathcal{L}}^\pm$ to the first
  integral
  \begin{align}
    \label{eq:24}
    -c \ol{u}_\pm + \frac{1}{2}\ol{\varphi^2_\pm} &= \frac{A}{2} ,
  \end{align}
  where we used $\bar{\mathcal{L}}^\pm[f''''] = 0$ by virtue of
  periodicity.  Equating the two relations in \eqref{eq:24} to
  eliminate $A$ gives
  \begin{align}\label{eq:TW_jump1}
    -c(\ol{u}_- - \ol{u}_+ ) + \frac{1}{2}\left(\ol{\varphi_-^2} -
    \ol{\varphi_+^2} \right)  & = 0,
  \end{align}
  which is the first jump condition \eqref{eq:RH1} when we identify
  $V = c$. Applying the operator $\bar{\mathcal{L}}^\pm$ to the second
  integral \eqref{eq:16} results in
  \begin{align}
    \label{eq:averaged_2nd_int}
    k_\pm^4 \ol{\varphi_\pm'''\varphi_\pm'} -  \frac{k_\pm^4}{2}
    \ol{(\varphi_\pm'')^2} + \frac{1}{6}\ol{\varphi_\pm^3}
    -\frac{c}{2}\ol{\varphi_\pm^2} - \frac{A}{2} \ol{u}_\pm + B = 0,
  \end{align}
  where we used
  $\bar{\mathcal{L}}^\pm[f'''f'] = k_\pm^4
  \ol{\varphi_\pm'''\varphi_\pm'}$.  Integrating by parts and applying
  \eqref{eq:TW_jump1} simplifies Eq.~\eqref{eq:averaged_2nd_int} to
  \begin{align}
  \label{eq:17}
    \frac{c}{2}\ol{\varphi_\pm^2} - \frac{1}{3}\ol{\varphi_\pm^3} -\frac{5}{2}
    k_\pm^4\ol{(\varphi_\pm'')^2} &=  -B . 
  \end{align}
  The third and final Rankine-Hugoniot condition \eqref{eq:RH2} with
  $V = c$ is found by subtracting the $+$ and $-$ instances of
  Eq.~\eqref{eq:17} to eliminate $B$
  \begin{align}
    \label{eq:TW_jump2}
    -\frac{c}{2}\left(\ol{\varphi_-^2} - \ol{\varphi_+^2} \right) +
    \frac{1}{3}\left(\ol{\varphi_-^3} - \ol{\varphi_+^3} \right) +
    \frac{5}{2}\left(k_-^4 \ol{(\varphi_-'')^2} - k_+^4
    \ol{(\varphi_{+}'')^2} \right) & = 0,
  \end{align}
  thereby completing the proof. 
\end{proof}
Theorem \ref{sec:whith-shocks:-abstr} motivates the following.
\begin{mydef}[admissibility]
  \label{sec:whith-shocks:-abstr-1}
  A KdV5 Whitham shock (recall
  Definition~\ref{sec:whith-shocks:-abstr}) is \textit{admissible} if
  there exists a heteroclinic traveling wave solution with speed
  $c = V$ to the KdV5 equation \eqref{eq:kdv5} satisfying the profile
  equation \eqref{eq:TW_ode} and the boundary conditions \eqref{eq:IC}
  corresponding to the left ($-$)/right ($+$) states of the Whitham
  shock \eqref{eq:BCs} and $V = c_+ = c_-$.
\end{mydef}

In the following section, we provide extensive numerical computations
of heteroclinic TW solutions that support the existence of admissible
Whitham shocks for the KdV5 equation.

\section{Whitham shocks}
\label{sec:WS}

In this section, we study admissible KdV5 Whitham shocks via
increasing levels of complexity. First, we consider the case of a
shock solution \eqref{eq:BCs} to the modulation equations where one
far-field state degenerates to zero wavenumber, i.e., a solitary
wave. These results are then generalized to Whitham shocks in which
two periodic waves satisfying the jump conditions propagate with
identical phase velocities. The section culminates with computations
of two co-propagating Whitham shocks where the corresponding TW
solutions are homoclinic, localized oscillatory patterns on an
oscillatory or uniform background.

\subsection{Solitary wave to periodic}
\label{sec:soli-to-per}
We consider the case where the left periodic wave degenerates to a
solitary wave ($k_- \to 0$) and the right periodic wave is of unit
amplitude and zero mean ($a_+ = 1$, $\ub_+ = 0$). The associated
Whitham shock \eqref{eq:BCs} is
\begin{equation}
  \label{eq:soli-per_Riemann}
  \big(\ub,a,k\big)(x,t) = \begin{cases}  (\ub_-, a_-, 0) & x < Vt \\  (0, 1 , k_+) & x \geq Vt \end{cases}.
\end{equation}
In the solitary wave limit ($k \to 0$) of the modulation equations
\eqref{eq:whitham_1}--\eqref{eq:whitham_3}, the averaged quantities
are
\begin{align*}
  \ol{\varphi^2} & = \ub^2, \quad \ol{\varphi^3} = \ub^3, \quad
                   \ol{(\varphi'')^2} = 0 .
\end{align*}
Therefore, the jump conditions \eqref{eq:RH1}--\eqref{eq:RH3} with a solitary wave on the left are 
\begin{align}\label{eq:soli-per_jump1}
-V \left(\ub_-\right) + \left(\frac{1}{2}\ub_-^2 - \frac{1}{2} \ol{\varphi_+^2} \right) & = 0,\\
-\frac{V}{2}\left(\ub_-^2 - \ol{\varphi_+^2} \right) + \left( \frac{1}{3}\ub_-^3 -  \frac{1}{3} \ol{\varphi_+^3} - \frac{5}{2}k_+^4 \ol{\varphi_{+,\theta\theta}^2}\right) & = 0, \\
V -c_+ = 0. \label{eq:soli-per_jump3}
\end{align}
An illuminating calculation using the Stokes wave approximation
\eqref{eq:stokes} for $\varphi_+$ leads to explicit formulae. In this
case, the jump conditions
\eqref{eq:soli-per_jump1}--\eqref{eq:soli-per_jump3} are solved by
\begin{align}\label{eq:weaklyNL_waveparams}
V = c_+ &= \frac{\ub_-}{2} - \frac{1}{16\ub_-}, \quad
k_+^4  = \frac{4\ub_-^3 + 3\ub_-}{15} - \frac{1}{80\ub_-}
\end{align}
where $\ub_-^2$ is a root of the polynomial 
\begin{align}\label{eq:ubar_quartic}
  1024 \left(\ub_-^2\right)^4 - 384 \left(\ub_-^2\right)^3 - 720 \left(\ub_-^2\right)^2 + 168\left(\ub_-^2\right) - 9 = 0 ,
\end{align}
which has three positive solutions
$\ub_-^2 \in \{0.08777, 0.1337, 0.9455\}$. The positive square roots
are inserted into \eqref{eq:weaklyNL_waveparams} to obtain the
wavenumber $k_+$, shock velocity $V$, and the right characteristic
velocities from Eqs.~\eqref{eq:30}--\eqref{eq:32} for three distinct
shock loci, all summarized in Table \ref{tab:stokes}. The left
solitary wave amplitude $a_-$ can be recovered from the solitary wave
speed-amplitude relation by equating it to the shock velocity
$c(\ub_-,a_-,0) = V$.  The left characteristic speeds for the left
solitary wave state are $\lambda_1^{(-)} = \ub_-$,
$\lambda_2^{(-)} = \lambda_3^{(-)} = V$.  The reason that two of the
characteristic velocities for the left state are the same is that two
of the three modulation equations coincide in the solitary wave limit
$k \to 0$, which is a well-known property of the Whitham equations
\cite{el_dispersive_2016-1}.  The cases where $\ub_-^2 < 0$ or
$\ub_- < 0$ in Eq.~\eqref{eq:ubar_quartic} can be dismissed because
these choices result in a complex value for $\ub_-$ or $k_+$ in
Eq.~\eqref{eq:weaklyNL_waveparams}.

\begin{table}[]
  \begin{center}
\begin{tabular}{l|r|l|c|c|c}
  \multicolumn{1}{c|}{$\ub_-$} & \multicolumn{1}{c|}{$V$} & \multicolumn{1}{c|}{$k_+$} & $\lambda_1^{(+)}$                   & $\lambda_2^{(+)}$                  & \multicolumn{1}{c}{$\lambda_3^{(+)}$} \\ \hline
  0.2963                    & $-0.0628 $         & 0.3936                  &   -0.7271 &                       0.1858 &   1.0416                \\
  0.3657                      & 0.0119              & 0.4775                  & -0.3768                                 &        0.4809               & 0.5360                               \\
  0.9724                      & 0.4219               & 0.8082                  & \multicolumn{1}{r|}{0.0586} & \multicolumn{1}{l|}{1.6775} & \multicolumn{1}{l}{2.5904}     
\end{tabular}
\end{center}
\caption{Three distinct Whitham shock loci and right ($+$)
  characteristic velocities for \eqref{eq:soli-per_Riemann} when the
  right periodic wave $\varphi_+$ is approximated by a weakly
  nonlinear Stokes wave \eqref{eq:stokes}, \eqref{eq:stokes_disp}.}
\label{tab:stokes}
\end{table}

\begin{table}[]
\begin{center}
\begin{tabular}{rl|r|l|c|c|c}
  & \multicolumn{1}{c|}{$\ub_-$} & \multicolumn{1}{c|}{$V$} & \multicolumn{1}{c|}{$k_+$} & $\lambda_1^{(+)}$                   & $\lambda_2^{(+)}$                  & \multicolumn{1}{c}{$\lambda_3^{(+)}$} \\ \hline
  $\blacktriangle$ & 0.2522                       & $-0.1133 $         & 0.3173              & \multicolumn{1}{r|}{$-0.1462$}                     &  \multicolumn{1}{r|}{$-0.0780$}                          &  \multicolumn{1}{r}{0.4056}                             
  \\
  $\newmoon$ & 0.3479                    & $-0.0071$             & 0.4496             &  \multicolumn{1}{r|}{0.6071 }                      & $-0.0237 - 0.1194 i$                  &  $ -0.0237 + 0.1194 i    $                           \\
  $\blacksquare$ & 0.9726                & 0.4213    & 0.8080               & \multicolumn{1}{r|}{0.0623} & \multicolumn{1}{r|}{1.6472} & \multicolumn{1}{r}{2.5723}  
\end{tabular}
\end{center}
\caption{Three distinct Whitham shock loci and right characteristic
  velocities using numerically computed periodic traveling waves.
  Compare with the approximate loci in Table \ref{tab:stokes}.}
\label{tab:numerics}
\end{table}
We test the accuracy of the approximate Whitham shock loci reported in
Table~\ref{tab:stokes} by solving the jump conditions
\eqref{eq:soli-per_jump1}--\eqref{eq:soli-per_jump3} with the family
of periodic traveling wave solutions obtained numerically in
Sec.~\ref{sec:stok-wave-appr}. Results from the numerical computations
are given in Table \ref{tab:numerics}.  The Whitham shock locus with
the largest root of Eq.~\eqref{eq:ubar_quartic} is well-approximated
to three digits (denoted by $\blacksquare$).  The corresponding right
characteristic velocities on this locus are accurate to one or two
digits.  The two Whitham shock loci corresponding to the two smaller
roots of Eq.~\eqref{eq:ubar_quartic} are less accurately approximated.
The reason for this is the Stokes approximation restriction
$k_+ \gg a_+^{1/4} = 1$ that requires a sufficiently large wavenumber
$k_+$.  The Whitham shock locus with the largest root $\ub_-$ also
exhibits the largest wavenumber $k_+$, hence is expected to be a
better approximation to the true Whitham shock locus, although good
agreement is notable given that $k_+ < 1$.  We remark that a numerical
search did not yield any other Whitham shock loci.

Two of the characteristic velocities for the locus denoted $\CIRCLE$
are complex, therefore we expect this locus to correspond to unstable
Whitham shocks.  We investigate stability questions in
Sec.~\ref{sec:two-shock-solns}.

\begin{figure}[h!]
  \begin{center}
    \includegraphics[scale=0.55]{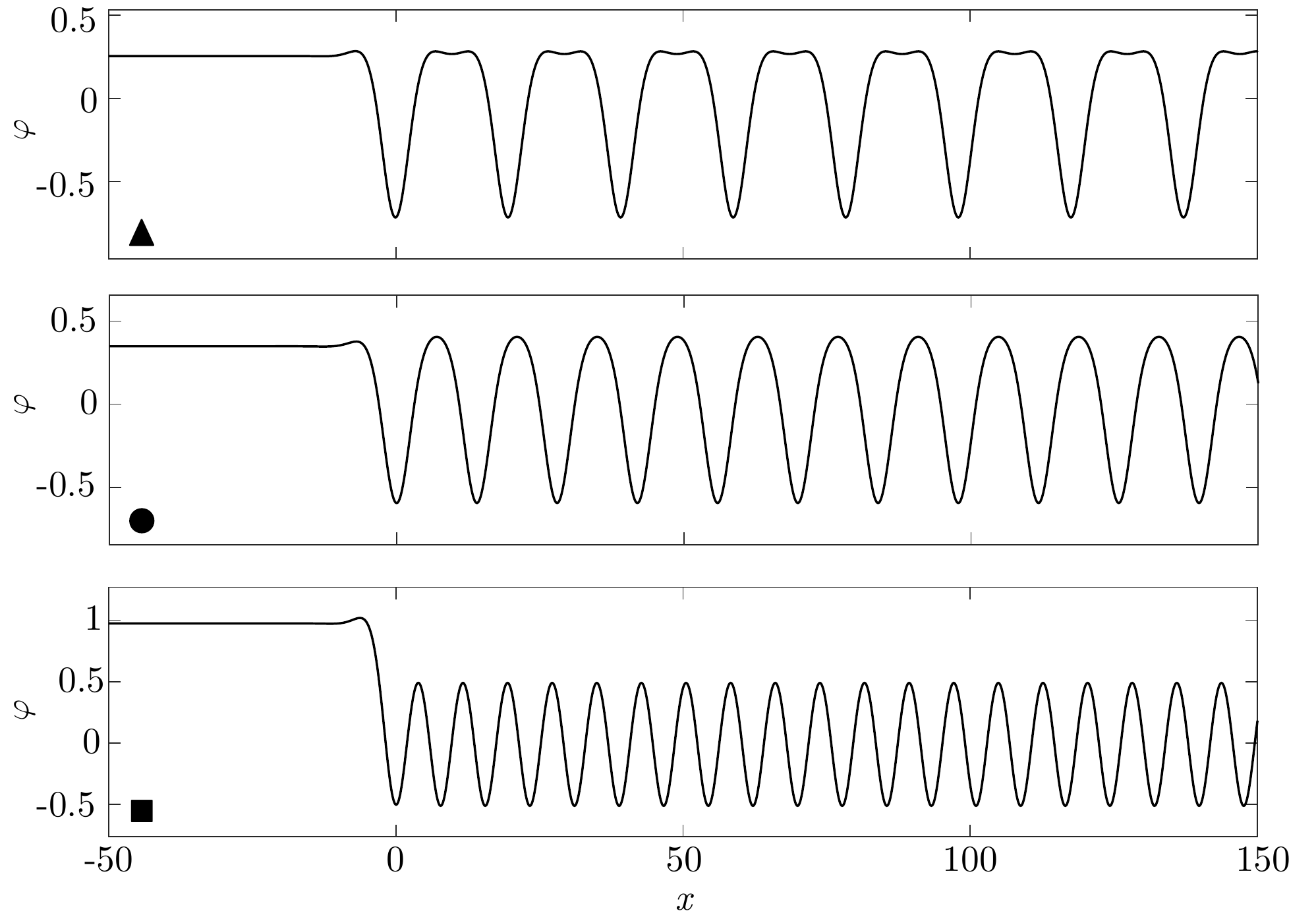}
    \caption{Computed traveling wave solutions corresponding to the
      periodic wave and mean values from the Whitham shock locus in Table
      \ref{tab:numerics}, matched by the symbols in the lower left
      corner. }
    \label{fig:WS}
  \end{center}
\end{figure}
We now compute heteroclinic traveling wave solutions whose zero
dispersion limit correspond to a Whitham shock from each of the loci
reported in Table~\ref{tab:numerics}, hence demonstrating Whitham
shock admissibility.  See the Appendix for computational details.  The
obtained solutions are depicted in Fig.~\ref{fig:WS}.  All three TW
solutions visually look similar for $x < 0$, which corresponds to the
left solitary wave $(\ub_-,a_-,0)$ of the Whitham shock
\eqref{eq:soli-per_Riemann}.  To investigate this further, we compare
this portion of the heteroclinic TW solutions with solitary wave
solutions that move with the shock velocity $V$ on the background
$\ub_-$ (the solitary wave amplitude $a_-$ is obtained from
$c(\ub_-,a_-,0) = V$). Figure \ref{fig:solitary_wave_comparison}
consists of the heteroclinic TW solutions depicted in
Fig.~\ref{fig:WS} overlaid with the left solitary wave (dashed red)
and the right periodic wave $\varphi(\theta;0,1,k_+)$ (dashed blue)
that form the corresponding admissible Whitham shock. Both the left
solitary waves and the right periodic waves are visually
indistinguishable from the heteroclinic TW in their respective regions
of validity. This corroborates our formulation and interpretation of
the zero dispersion limit of heteroclinic TW solutions as admissible
Whitham shocks where the left wave is a solitary wave that rapidly
transitions to a co-propagating finite amplitude periodic wave.

\begin{figure}[h!]
\begin{center}
\includegraphics[scale=0.45]{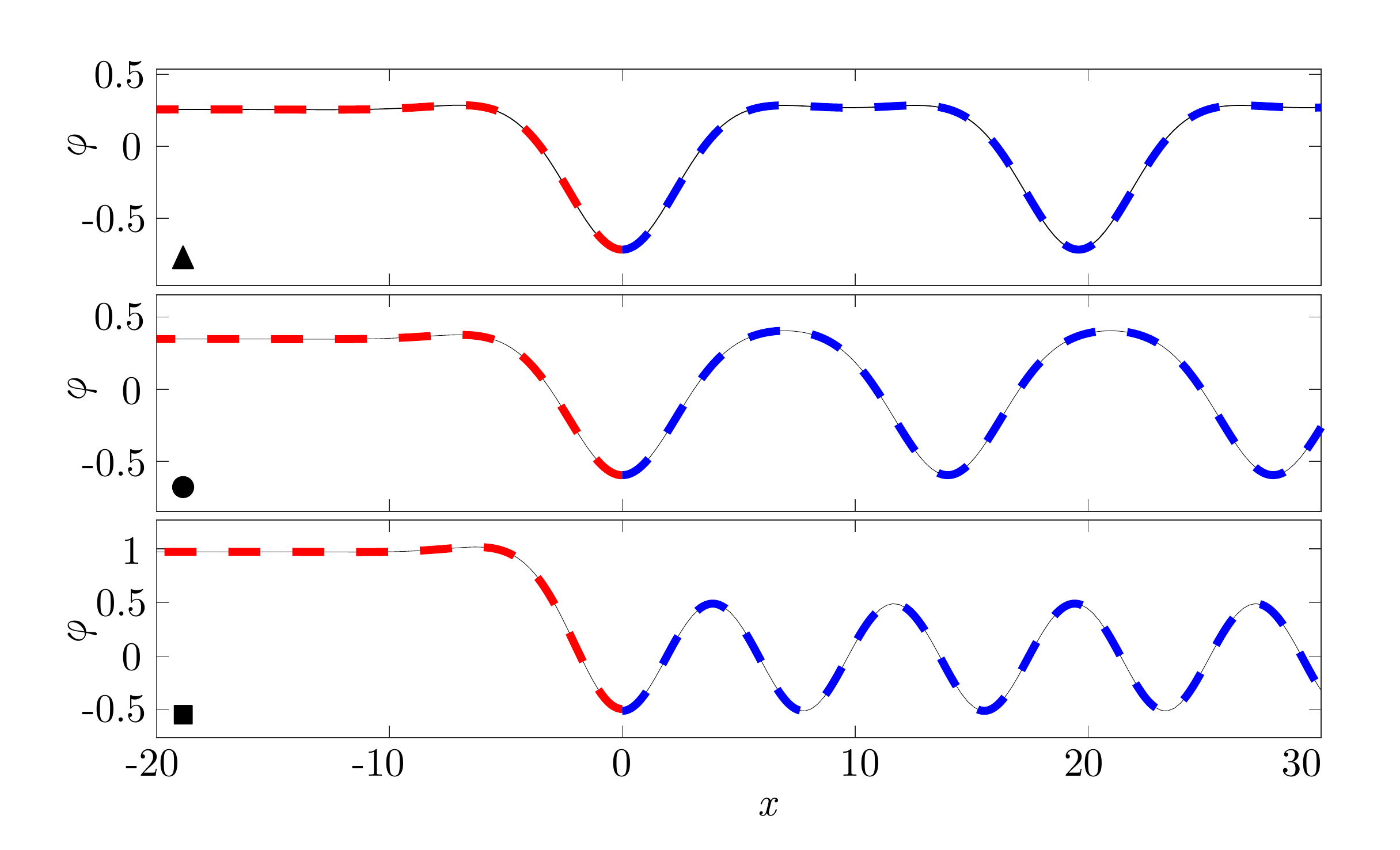}
\caption{Whitham shock structure via heteroclinic TW solutions
  (solid).  The left solitary wave (dashed red for $x < 0$) and right
  periodic wave (dashed blue for $x > 0$) from the Whitham shock loci
  reported in Table \ref{tab:numerics} are overlayed on the
  heteroclinic TW solution.  The symbols in the lower left corners
  coincide with those in Fig.~\ref{fig:WS} and Table
  \ref{tab:numerics}. }
\label{fig:solitary_wave_comparison}
\end{center}
\end{figure}

Admissible Whitham shock solutions and characteristics
\begin{equation}
  \label{eq:7}
  \Gamma_j = \left \{(x,t) ~ | ~ \frac{\mathrm{d}x}{\mathrm{d}t} =
    \lambda_j \right \}, \quad j = 1,2,3,
\end{equation}
corresponding to the loci $\blacktriangle$ and $\blacksquare$ in Table
\ref{tab:numerics} are shown in Figures \ref{fig:WS_chars1} and
\ref{fig:WS_chars2}, respectively.  Both Whitham shocks with real
characteristic velocities are weakly compressive in the first
characteristic family $\Gamma_1$ because
$\lambda_1^{(+)} < V = \lambda_1^{(-)}$, while the second
characteristic family is weakly expansive
$\lambda_2^{(-)} = V < \lambda_2^{(+)}$, and the third characteristic
family passes through the Whitham shock, decelerating
$V < \lambda_3^{(+)} < \lambda_3^{(-)}$ for $\blacktriangle$ and
accelerating $V < \lambda_3^{(-)} < \lambda_3^{(+)}$ for
$\blacksquare$ \cite{dafermos_hyperbolic_2009}.  Consequently, we
refer to this class of Whitham shock solutions as \textit{weak
  1-shocks}.

The degeneration of the periodic wave to a solitary wave on the left
state in Fig.~\ref{fig:WS} allows us to obtain three additional
admissible Whitham shock loci by reflecting the initial data
\eqref{eq:soli-per_Riemann} and the heteroclinic traveling wave across
$x = 0$. These solutions are weakly compressive in the second
characteristic family, $\Gamma_2$, hence are called \textit{weak
  2-shocks}. Further implications of this observation are discussed in
Sec.~\ref{sec:two-shock-solns}.

\begin{figure}[h!]
\begin{center}
\includegraphics[scale=0.4]{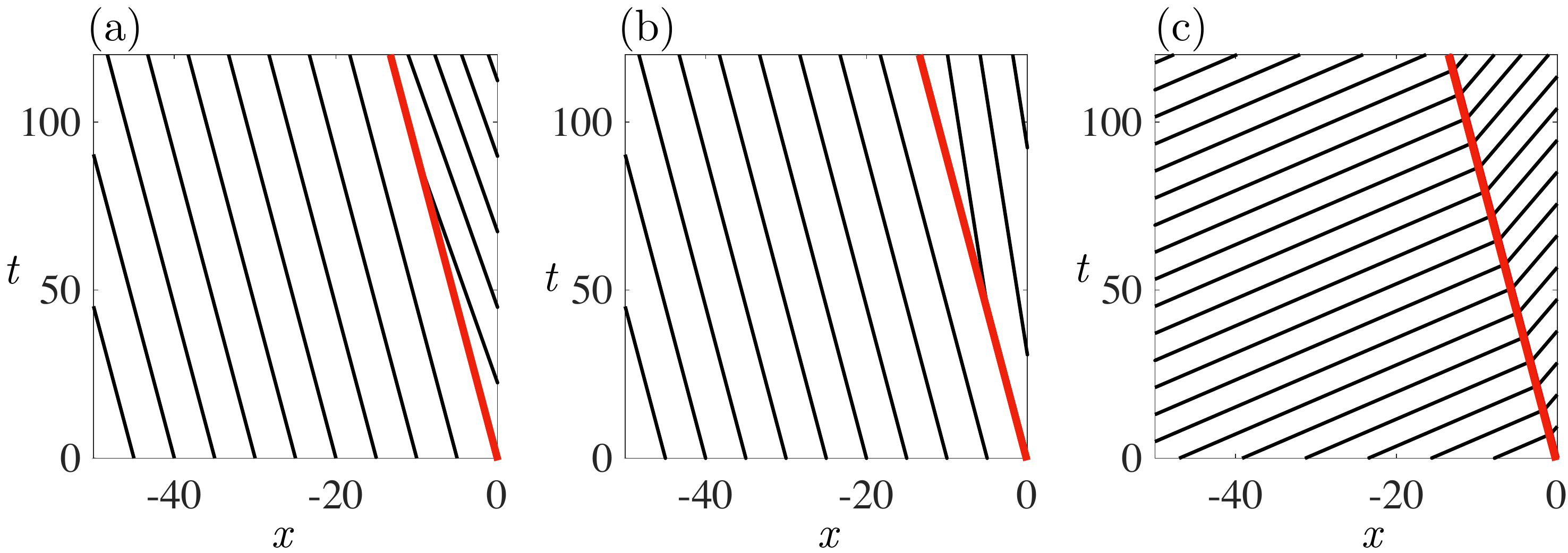}
\end{center}
\caption{Characteristics of the KdV5-Whitham modulation equations for
  the Whitham shock locus $\blacktriangle$ in Table
  \ref{tab:numerics}. (a) Weakly compressive 1-wave characteristics
  $\Gamma_1$, (b) Weakly expansive 2-wave characteristics $\Gamma_2$,
  and (c) Decelerating 3-wave characteristic family $\Gamma_3$. The
  shock is identified by the thick, red line. }
\label{fig:WS_chars1}
\end{figure}

\begin{figure}[h!]
\begin{center}
\includegraphics[scale=0.4]{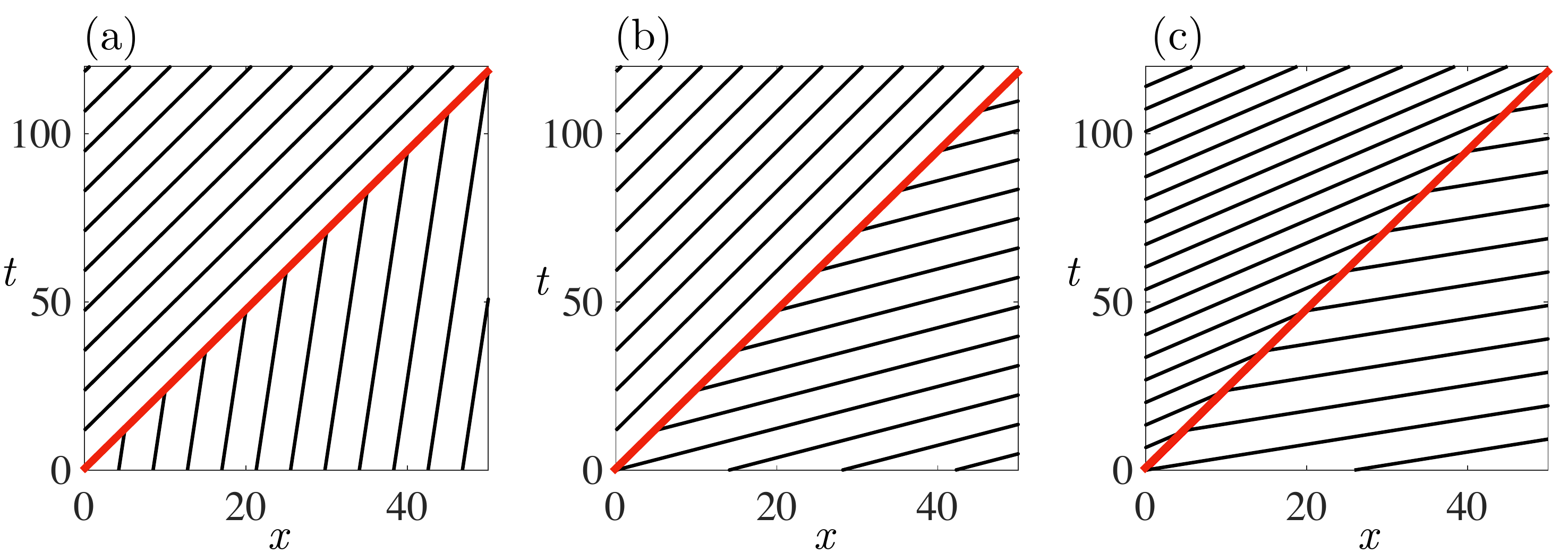}
\end{center}
\caption{Characteristics of the KdV5-Whitham modulation equations for
  the solution $\blacksquare$ in Table \ref{tab:numerics}. (a) Weakly
  compressive 1-wave characteristics $\Gamma_1$, (b) Weakly expansive
  2-wave characteristics $\Gamma_2$, and (c) Accelerating 3-wave
  characteristic family $\Gamma_3$. The shock is identified by the
  thick, red line.}
\label{fig:WS_chars2}
\end{figure}

\subsection{Periodic to periodic}
\label{sec:per-per}
Let us continue our computation of the jump conditions
Eqs.~\eqref{eq:RH1}--\eqref{eq:RH3} where we now consider shock
solutions of the modulation equations corresponding to one periodic
wave connected to another. We slightly modify the Riemann problem
\eqref{eq:soli-per_Riemann} from the previous section where the left
state is assumed to have zero wavenumber. In this section, we scale
the Whitham shock \eqref{eq:BCs} so that the left state consists of an
arbitrary periodic wave, and the right state consists of a periodic
wave with zero mean, unit amplitude and arbitrary wavenumber. This
parameter set results in the shock solution
\begin{align}\label{eq:per-per_riemann}
  (\ub,a,k)(x,t) = \begin{cases} (\ub_-,a_-,k_-) &x < Vt \\ (0,1,k_+)
    & x \geq  Vt \end{cases},
\end{align}
where $V = c_\pm$. The Rankine Hugoniot jump relations
\eqref{eq:RH1}--\eqref{eq:RH3} are a nonlinear system of three
equations that relate the four remaining parameters: $k_+$, $\ub_-$,
$a_-$, and $k_-$. We use $\ub_-$ as the continuation parameter to
obtain the one-parameter family of Whitham shock loci
\begin{align*}
  k_+ = k_+(\ub_-), \quad a_- = a_-(\ub_-), \quad k_- = k_-(\ub_-) .
\end{align*}

Attempts to solve the jump conditions \eqref{eq:RH1}--\eqref{eq:RH3}
approximately by using the Stokes wave approximation \eqref{eq:stokes}
yield no nontrivial asymptotically valid solutions, aside from the
solitary wave to periodic shocks discussed previously. As a result, we
rely on numerical continuation along the parameterized curves
$(k_+,a_-,k_-)(\ub_-)$ where we start from the three Whitham shock
loci calculated in Section \ref{sec:soli-to-per} so that $k_- = 0$ and
$\ub_-$, $k_+$, $V$ are initialized from Table \ref{tab:numerics}
($a_-$ satisfies $c(\ub_-,a_-,0) = V$). We numerically continue
solutions to the jump conditions \eqref{eq:RH1}--\eqref{eq:RH3} with
Matlab's fsolve function for decreasing values of
$\ub_-$. Continuation is terminated when the value
$k_+(\ub_-) < 10^{-3}$ is reached, indicating that the oscillatory
wavetrain in the left far-field is nearly a solitary wave. Figures
\ref{fig:per-per_phase_diag}(a), (c) and (e) present the three Whitham
shock loci.  We also compute heteroclinic TW solutions for select
$\ub_-$, seeded with left/right states from Whitham shock loci, and
plot them in Figures \ref{fig:per-per_phase_diag}(b), (d), and (f).
Consequently, we find that all three periodic to periodic Whitham
shock loci are admissible.

\begin{figure}
\begin{center}
\includegraphics[scale=0.45]{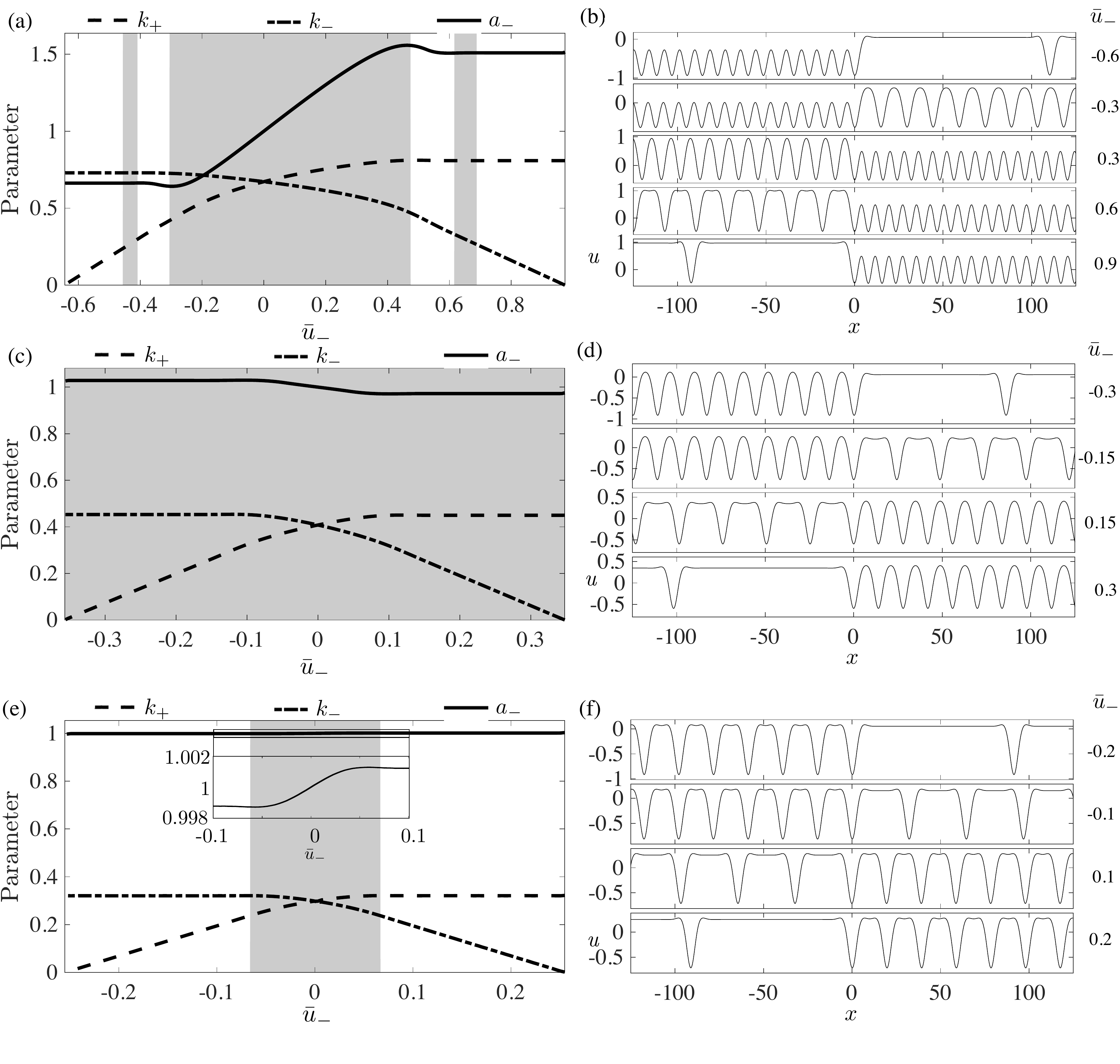}
\end{center}
\caption{Continuation curves of admissible periodic to periodic
  Whitham shock solutions \eqref{eq:per-per_riemann} and corresponding
  heteroclinic traveling wave solutions.  (a),(c),(e) Shock loci for
  the periodic wave parameters $k_-(\ub_+)$, $a_+(\ub_+)$,
  $k_+(\ub_+)$. (b), (d), (f) Example heteroclinic solutions
  corresponding to the shock curves (a), (c), and (e)
  respectively. Gray, shaded areas in (a), (c) and (e) correspond to
  complex characteristic velocities and modulational instability. The
  zoomed-in inset in Figure (e) demonstrates that the amplitude is not
  constant across the solution curve. }
\label{fig:per-per_phase_diag}
\end{figure}


\begin{figure}
\begin{center}
\includegraphics[scale=0.3]{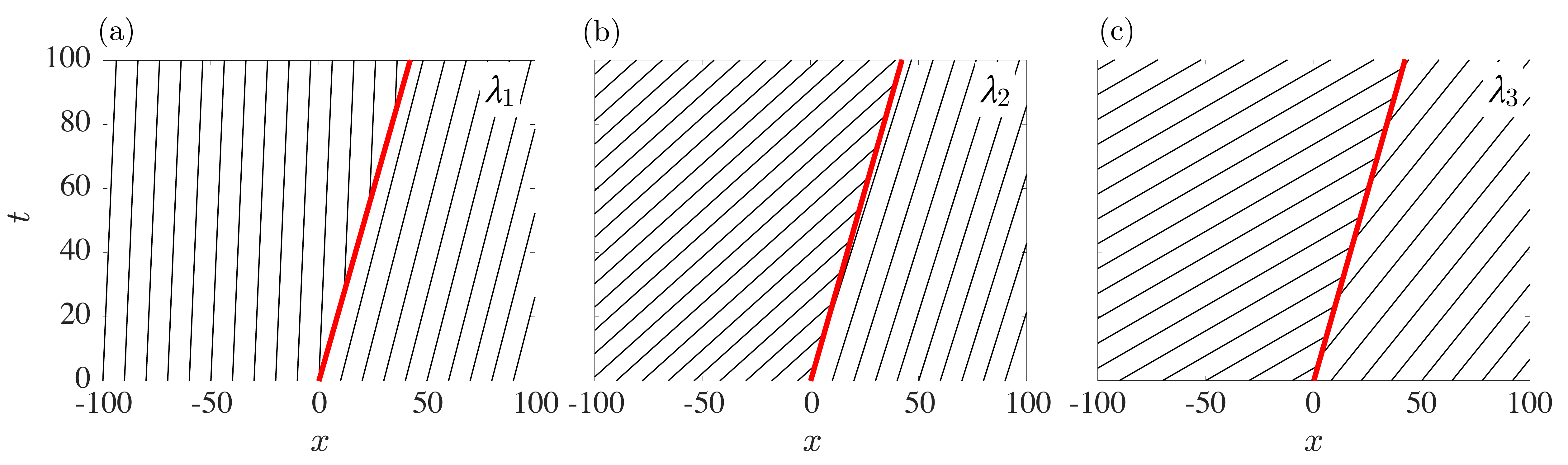}
\caption{(a)-(c) Characteristic families of the undercompressive
  Whitham shock in Fig.~\ref{fig:per-per_phase_diag}(b) with
  $\ub_+ = 0.6$ where
  $\lambda_1^{(\pm)} < V < \lambda_2^{(\pm)} < \lambda_3^{(\pm)}$. The
  shock trajectory is depicted with the thick red curve.}
\label{fig:perper_characteristics}
\end{center}
\end{figure}

Dynamic stability of admissible Whitham shocks is determined by the
hyperbolicity of the Whitham modulation equations (recall
Fig.~\ref{fig:characteristics}). Whitham shocks with real
characteristic velocities are identified by the white background in
Fig.~\ref{fig:per-per_phase_diag} (a), (c), and (e) while the gray
background denotes complex characteristic velocities
$\mathrm{Im} \lambda_1 \ne 0$ and $\mathrm{Im} \lambda_2 \ne 0 $ and
therefore a modulationally unstable heteroclinic TW.  This will be
investigated further in Sec.~\ref{sec:two-shock-solns}.

In Figure \ref{fig:perper_characteristics}, we plot the numerically
computed characteristic curves $\Gamma_j$, $j = 1,2,3$ and the shock
trajectory for the periodic to periodic Whitham shock in
Fig.~\ref{fig:per-per_phase_diag}(a) where $\ub_- = 0.6$ (the
corresponding heteroclinic TW is shown in
Fig.~\ref{fig:per-per_phase_diag}(b)).  We observe that all
characteristic families pass through the shock and satisfy
\begin{equation}
  \label{eq:8}
  \lambda_1^{(\pm)} < V < \lambda_2^{(\pm)} < \lambda_3^{(\pm)} .
\end{equation}
Shocks that are not compressive in any characteristic family violate
the Lax entropy conditions.  In the conservation law community,
non-Lax shocks can be identified as admissible when they are the limit
of vanishing dissipative-dispersive heteroclinic TWs
\cite{jacobs_travelling_1995,lefloch_hyperbolic_2002}.  In such cases,
they are referred to as \textit{undercompressive}.  Here, we find that
the vanishing dispersion limit of heteroclinic periodic to periodic
TWs generally converge to undercompressive Whitham shocks.  To prove
this, consider Fig.~\ref{fig:characteristic_phase_velocities} where
the scaled characteristic velocities $(\lambda_j - \ub)/a$, $j = 1,2$
and shock velocity $(V - \ub)/a$ are plotted in the scaled coordinate
$\tilde{k} = a^{-1/4}k$.  Since $\lambda_3 > \lambda_2$, we see that
all admissible Whitham shocks with real characteristic velocities
exhibit the undercompressive relations \eqref{eq:8} except for the
very narrow band of waves $0.662 \lesssim k a^{-1/4} \lesssim 0.672$.
Because this band is so narrow in parameter space, we do not
investigate these solutions any further.  Undercompressive shocks were
first described for general $2 \times 2$ systems in
\cite{shearer_solution_1987} and have since been observed, for
example, in fluid dynamics \cite{bertozzi_undercompressive_1999}.

\begin{figure}[h!]
\begin{center}
\includegraphics[scale=0.45]{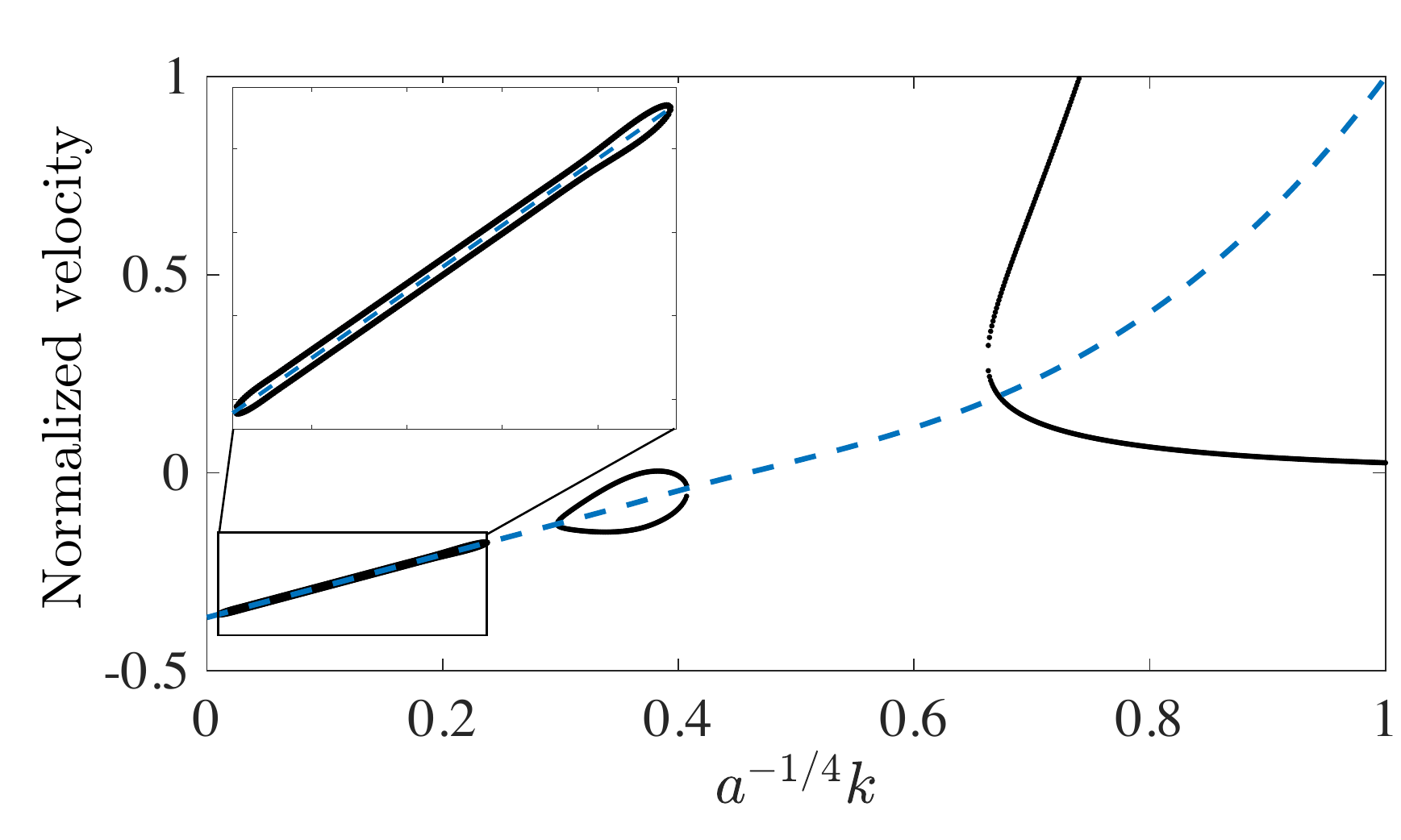}
\caption{Numerically computed normalized Whitham shock velocity
  $(V-\ub)/a$ (blue, dashed) and normalized, purely real
  characteristic velocities $(\lambda_1-\ub)/a < (\lambda_2-\ub)/a$
  (black dots). The inset is a zoomed in view of the phase velocity
  and characteristics velocities for small $a^{-1/4}k$.}
\label{fig:characteristic_phase_velocities}
\end{center}
\end{figure}

\subsection{Two shock solutions}
\label{sec:two-shock-solns}

In Sections \ref{sec:soli-to-per} and \ref{sec:per-per}, we
established the admissibility of KdV5-Whitham shocks by computing
heteroclinic traveling wave solutions to the KdV5 equation
\eqref{eq:kdv5}.  Each locus of admissible Whitham shocks depicted in
Fig.~\ref{fig:per-per_phase_diag}(a), (c), (e) include, up to scaling
and Galilean symmetries, pairs of shocks with reflected $+$ and $-$
states. The implication of this observation is that it is possible to
arrange for co-propagating Whitham shocks with the same equilibrium on
the left and right, respectively, as well as the periodic wave on the
right and left, respectively.  Can we compute homoclinic TWs
consisting of two co-propagating Whitham shocks? Furthermore, can this
construction be extended to solutions tending to a periodic wave in
the far-field?  We affirmatively answer these questions now.
 
We formulate the double Whitham shock problem as a locally periodic
wave that terminates after a finite number of oscillation periods and
transitions to a different periodic wave, possibly a solitary wave, in
the far field. To this end, let us consider solutions of the form
\begin{align}
  \label{eq:riemann_2shock}
  (\ub,a,k)(x,t) = \begin{cases}
    (0,1,k_{\rm i}) & |x - Vt| \leq \frac{n\pi}{k_{\rm i}}\\
    (\ub_{\rm o},a_{\rm o},k_{\rm o}) & |x- Vt| \geq \frac{n\pi}{k_{\rm i}}
  \end{cases},
\end{align}
where the subscript i refers to the inner periodic wave with $n$
oscillations and subscript o denotes the outer periodic wave in the
far field. Following the normalization in Sec.~\ref{sec:per-per}, we
have normalized the inner periodic solution to have zero mean and unit
amplitude.

Figure \ref{fig:2_shock} depicts the structure of an example
homoclinic TW solution in the physical plane overlaid with
co-propagating weak 1-shock and weak 2-shock solutions that satisfy
the jump conditions for zero wavenumber in the far field, i.e.,
$k_{\rm o} =0$. In general, the far-field periodic wave can be chosen
from a continuous set of nonzero wavenumbers
(cf.~Fig.~\ref{fig:per-per_phase_diag} for the range of possible
values, equating $k_{\rm o} = k_-$).
\begin{figure}[h!]
  \begin{center}
    \includegraphics[scale=0.4,page=1]{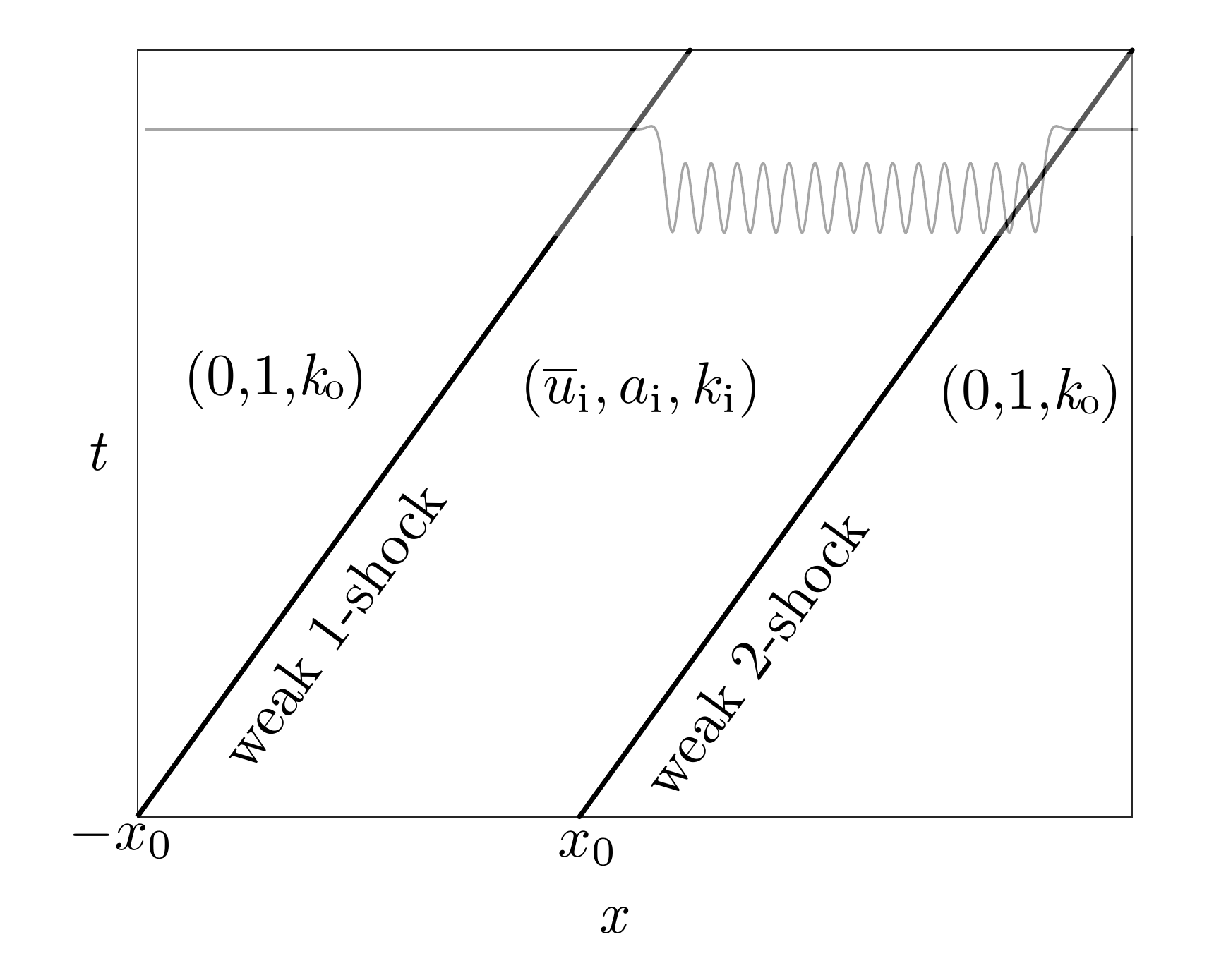}
    \caption{Double Whitham shock solution satisfying the jump
      conditions \eqref{eq:RH1}--\eqref{eq:RH3} from the interior
      periodic wave $(0,1,k_{\rm i})$ to the outer solitary wave
      $(\ub_{\rm o},a_{\rm o},0)$. The shocks are identified by the
      solid black curves and the solution in physical space is shown
      in gray.}
    \label{fig:2_shock}
  \end{center}
\end{figure}

Utilizing the three Whitham shock loci in Table \ref{tab:numerics},
the scaling symmetries \eqref{eq:19} and \eqref{eq:20}, and double
Whitham shock data \eqref{eq:riemann_2shock}, we compute homoclinic TW
profiles and initialize the KdV5 equation with these solutions plus
small, band-limited noise (see Appendix for computational details).
The time evolution of the perturbed homoclinic TW solutions then
allows us to corroborate our hypothesis that traveling waves comprised
of modulationally unstable portions will be unstable.  Indeed, in
Fig.~\ref{fig:localized_2shock}, we depict the evolution of three
distinct homoclinic solutions with $n = 25$ oscillation periods.  Note
that our computational approach enables for the construction of
homoclinic solutions with any number of oscillation periods. The
intermediate periodic wave in Fig.~\ref{fig:localized_2shock}(b) lies
on the Whitham shock locus $\CIRCLE$ in Table \ref{tab:numerics},
which exhibits complex characteristic velocities.  As expected, the
homoclinic TW solution is unstable. In
Figs.~\ref{fig:localized_2shock} (a) and (c), the homoclinic
oscillatory defect is comprised of periodic waves on the Whitham shock
loci $\blacktriangle$ and $\blacksquare$ in Table \ref{tab:numerics},
respectively, where the modulation equations are strictly
hyperbolic. Perturbed homoclinic traveling waves are numerically
stable over a long integration time.

These homoclinic TWs are similar to solutions obtained in the
investigation of reversible Hamiltonian systems
\cite{buffoni_bifurcation_1996,champneys_homoclinic_1998}, though we
have presented an alternative approach by which multi-pulse solutions
with an arbitrary number of peaks can be computed from the Whitham
shock loci in Table \ref{tab:numerics}.

\begin{figure}
\begin{center}
\includegraphics[scale=0.25,page = 1]{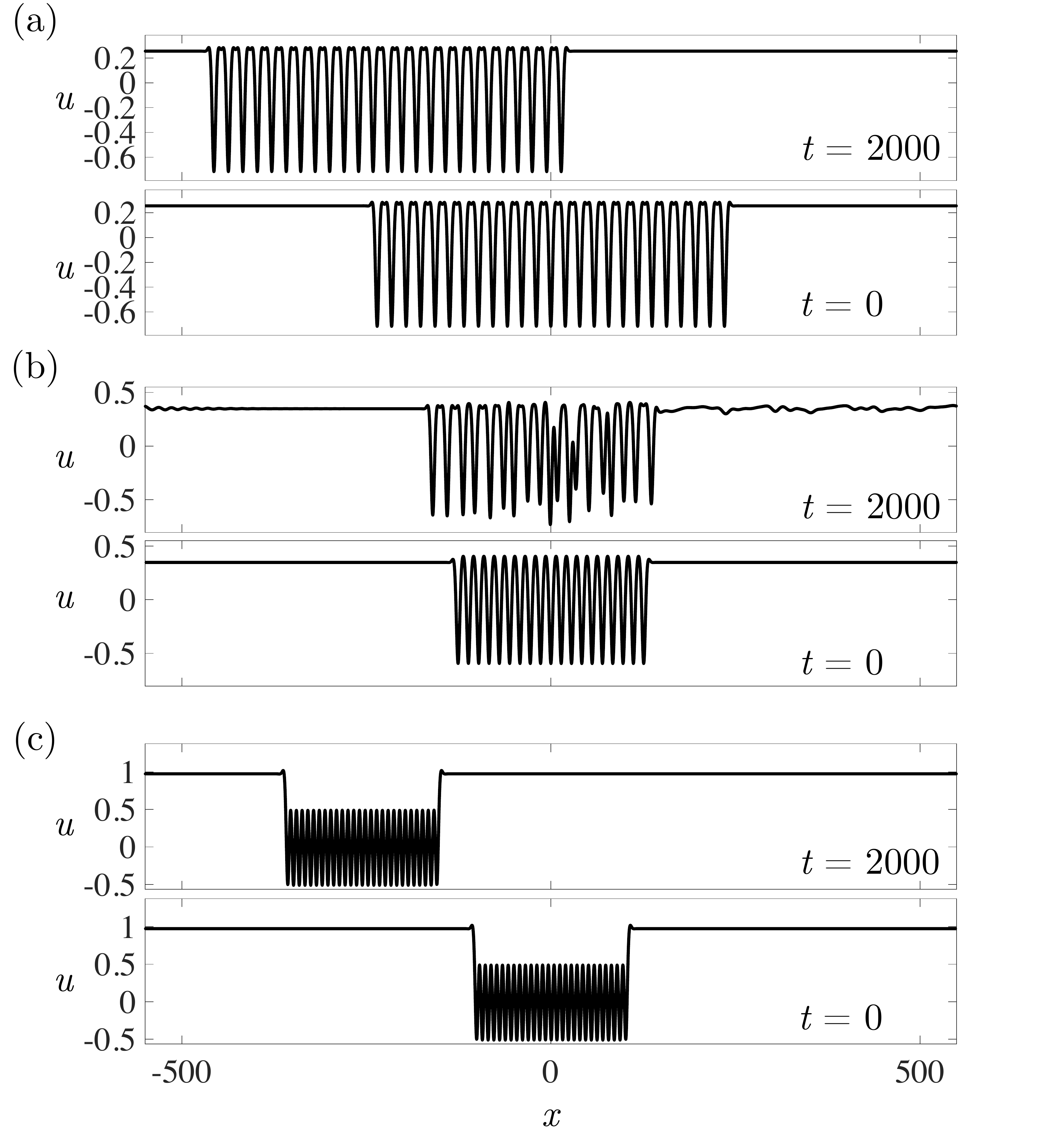}
\end{center}
\caption{Evolution of perturbed homoclinic TW solutions whose
  parameters in Eq.~\eqref{eq:riemann_2shock} with $n = 25$ lie on the
  Whitham shock loci in Table \ref{tab:numerics}. Initial data at
  $t = 0$ is evolved to $t = 2000$. (a) Double Whitham shocks
  $\blacktriangle$ in Table \ref{tab:numerics} with real
  characteristics.  (b) Double Whitham shocks $\CIRCLE$ with complex
  characteristics. (c) Double Whitham shocks $\blacksquare$ with real
  characteristics.  }
\label{fig:localized_2shock}
\end{figure}

Similarly, homoclinic TWs consisting of distinct inner and outer
periodic waves can be constructed from Whitham shock loci shown in
Fig.~\ref{fig:per-per_phase_diag}. Numerical computations of initially
perturbed homoclinic traveling waves of this type along with their
long time evolution are displayed in Fig.~\ref{fig:perper_2shock},
again revealing numerically stable evolution for those solutions with
all real characteristic
velocities---Fig.~\ref{fig:perper_2shock}(a,d)---and numerically
unstable evolution for those double Whitham shocks that exhibit
complex characteristic velocities---Fig.~\ref{fig:perper_2shock}(b,c).

\begin{figure}
\begin{center}
\includegraphics[scale=0.45,page = 1]{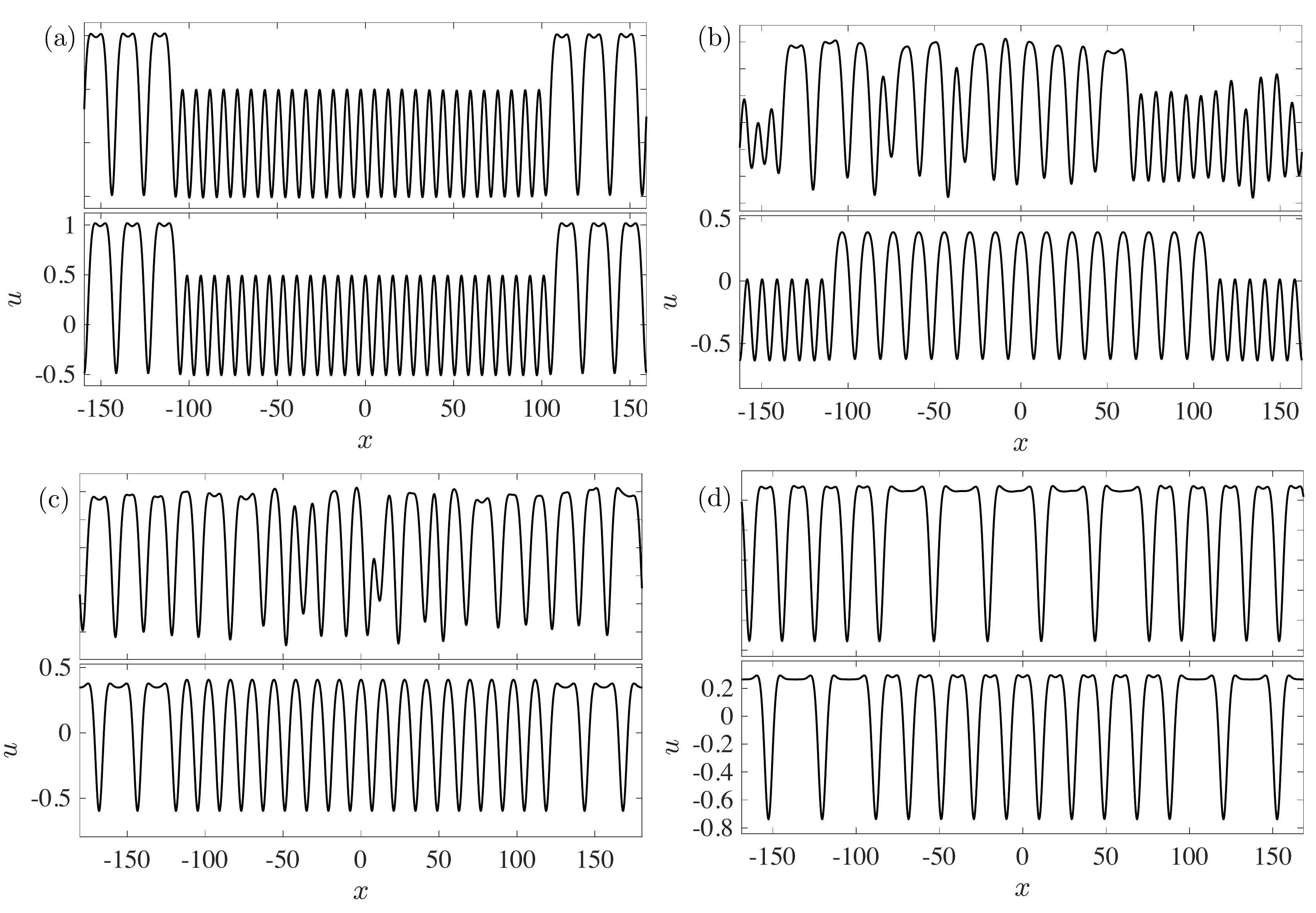}
\end{center}
\caption{Evolution of perturbed TWs corresponding to co-propagating
  Whitham shocks shown at initial time $t = 0$ and final time
  $t = 1500$. Parameter values are selected from those depicted in
  Figure \ref{fig:per-per_phase_diag}. (a) TW with parameters in
  Fig. \ref{fig:per-per_phase_diag}(a) with
  $\ub_{\rm o } = \ub_- = 0.6$.  (b) TW with parameters in
  Fig. \ref{fig:per-per_phase_diag}(a) with
  $\ub_{\rm o } = \ub_- = -0.3$. (c) TW with parameters
  \ref{fig:per-per_phase_diag}(c) with $\ub_{\rm o } = \ub_- = 0.15$.
  (d) TW with parameters in Fig. \ref{fig:per-per_phase_diag}(e) with
  $\ub_{\rm o } = \ub_- = 0.1$.  }
\label{fig:perper_2shock}
\end{figure}

\section{Applications and extensions}
\label{sec:applications}
In this section, we discuss applications of the theory presented in
this manuscript and various physical systems to which it applies. We
construct the modulation solution for the TDSW depicted in
Fig.~\ref{fig:TDSW_intro} for the KdV5 equation \eqref{eq:kdv5} and
discuss how the same construction can be carried out for the Kawahara
equation \eqref{eq:kawahara}, where these structures were numerically
coputed in \cite{sprenger_shock_2017}. This section concludes with
numerical experiments simulating the Riemann problem
\eqref{eq:riemann_data} for the Whitham equation
\cite{whitham_linear_1974,whitham_variational_1967}, which describes
weakly nonlinear, fully dispersive gravity-capillary water waves.

\subsection{Dam break problem for the modulation equations: Traveling dispersive shock waves}
\label{sec:kdv5_TDSW}

Several recent works have considered step initial data for systems
with higher order dispersion
\cite{sprenger_shock_2017,hoefer_modulation_2019,conforti_resonant_2014,smyth_dispersive_2016,el_radiating_2016}. Among
each of these studies, one key observation from numerical simulations
was made: the nonlinear structure arising from smoothed step initial
data developed into a wave with two distinct regions, depicted in
Fig.~\ref{fig:TDSW_intro}. We now identify the two distinct regions as
follows. At the left, trailing edge of the TDSW, there is a
heteroclinic traveling wave corresponding to an admissible Whitham
shock solution of the modulation equations
(cf.~Sec.~\ref{sec:soli-to-per}). The heteroclinic traveling wave
terminates at a partial DSW described by a continuous, simple wave
solution of the modulation equations
(cf.~Sec.~\ref{sec:strongly-nonl-regime}) at the leading edge. Because
the modulated leading edge initiates from an oscillatory wavetrain,
rather than a solitary wave, it is termed a partial DSW. Similar
partial DSWs have been observed in the initial boundary value problem
of nonlinear systems with lower order dispersion
\cite{hoefer_piston_2008,marchant_initial_2002}. We now explain this
Whitham shock-rarefaction modulation solution to the intial value
problem
\begin{align}\label{eq:kdv5_riemann}
  u_t + uu_x + u_{xxxxx} = 0 , \qquad u(x,0) = \begin{cases}1 & x
    < 0 \\ 0  & x \geq 0 \end{cases}.
\end{align}

Numerical simulations of \eqref{eq:kdv5_riemann} in
Fig.~\ref{fig:TDSW_numerics} show two distinct regions in the
solution. The boundaries of these regions are defined by identifying
the following velocities: the trailing edge Whitham shock velocity,
$V_{\rm s}$, the intermediate velocity $V_i$ joining the TW to the
partial DSW and $V_+$ the leading, harmonic edge velocity where the
partial DSW terminates when the modulation amplitude vanishes
$a \to 0$. If we denote the parameters of the periodic portion of the
heteroclinic wave by $(\ub_r,a_r,k_r)$ and the leading harmonic wave
edge wavenumber by $k_+$, then $V_i = \lambda_2(\ub_r,a_r,k_r)$ and
$V_+ = \lambda_2(0,0,k_+)$.  The TDSW traveling wave portion bears a
striking resemblance to the traveling wave that is identified by the
symbol $\blacksquare$ in Fig.~\ref{fig:WS}. The Whitham shock locus
for this solution in Table \ref{tab:numerics} is well approximated
using a Stokes wave in Table \ref{tab:stokes}, scaled by
\eqref{eq:19}, \eqref{eq:20} so that $\ub_- \to 1$ and
$\ub_+$ is a to-be-determined parameter. The intermediate periodic
wave properties are then given as a one parameter family
\begin{align*}
  a_{\rm r} = \frac{1}{\zeta}(1-\ub_{\rm r}), \quad k_{\rm r}^4 =
  \left[\frac{4\zeta^2 + 3}{15} - \frac{1}{80\zeta^2}
  \right](1-\ub_{\rm r}) \quad V_s = \left[ \frac{1}{2} -
  \frac{1}{16\zeta^2} \right](1-\ub_{\rm r}) + \ub_{\rm  r},
\end{align*}
where $\zeta \approx 0.9724$ is the positive square root of the
largest, positive real root of the quartic polynomial
\eqref{eq:ubar_quartic}, and $\ub_{\rm r}$ is the mean to be determined
by matching to a simple wave describing the transition from the
oscillatory wavetrain to the leading, constant level, $\ub = 0$
propagating with velocity $V_+$. The computation of the self-similar
simple wave curve is described in Sec.~\ref{sec:strongly-nonl-regime},
denoted by $\mathbf{q}(x/t)$. The Whitham shock locus, 2-wave curve
and the far-field mean values completely determine the Whitham
shock-rarefaction modulation solution
\begin{align*}
  (\ub,a,k)(x,t) = \begin{cases}(1,a_-,0) &  V_{\rm s} t < x \\
    (\ub_{\rm r},a_{\rm r},k_{\rm r}) & V_{\rm s} t \leq x < V_{\rm i} t\\
    \mathbf{q}(x/t)& V_{\rm i }t  \leq x < V_{+} t\\
    (0,0,k_{+}) &  V_+ t \leq x
\end{cases}.
\end{align*}
This shock-rarefaction solution can be viewed as the higher dispersion
analog of the shock-rarefaction solution for the shallow water dam
break problem or the gas dynamics piston problem.

\begin{figure}
  \begin{center}
    \includegraphics[scale=0.3]{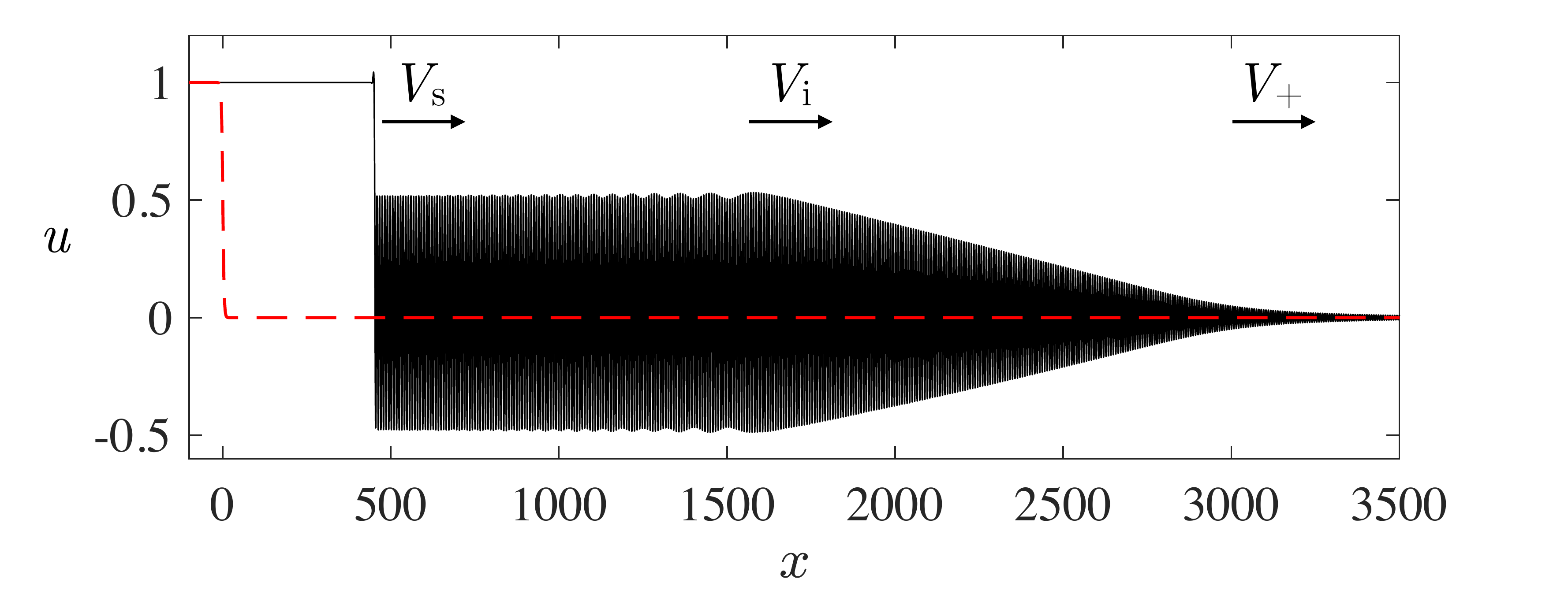}
    \label{fig:TDSW_numerics}
  \end{center}
  \caption{Numerical evolution of the initial value problem
    \eqref{eq:kdv5_riemann} with smoothed step data (dashed red).
    Three distinct velocities are noted. }
\end{figure}

In Figure \ref{fig:tdsw_numerics_comparison}, we compare the Whitham
shock-rarefaction solution of the modulation equations with numerical
simulations of smoothed step initial data.  Numerically extracted
values of the modulation parameters exhibit excellent agreement with
the Whitham shock-rarefaction modulation solution. The small
oscillations in the modulation variables that occur along the
intermediate equilibrium state are due to higher order dispersive
effects that are not captured by leading order modulation theory.

\begin{figure}[h!]
  \includegraphics[scale=0.4]{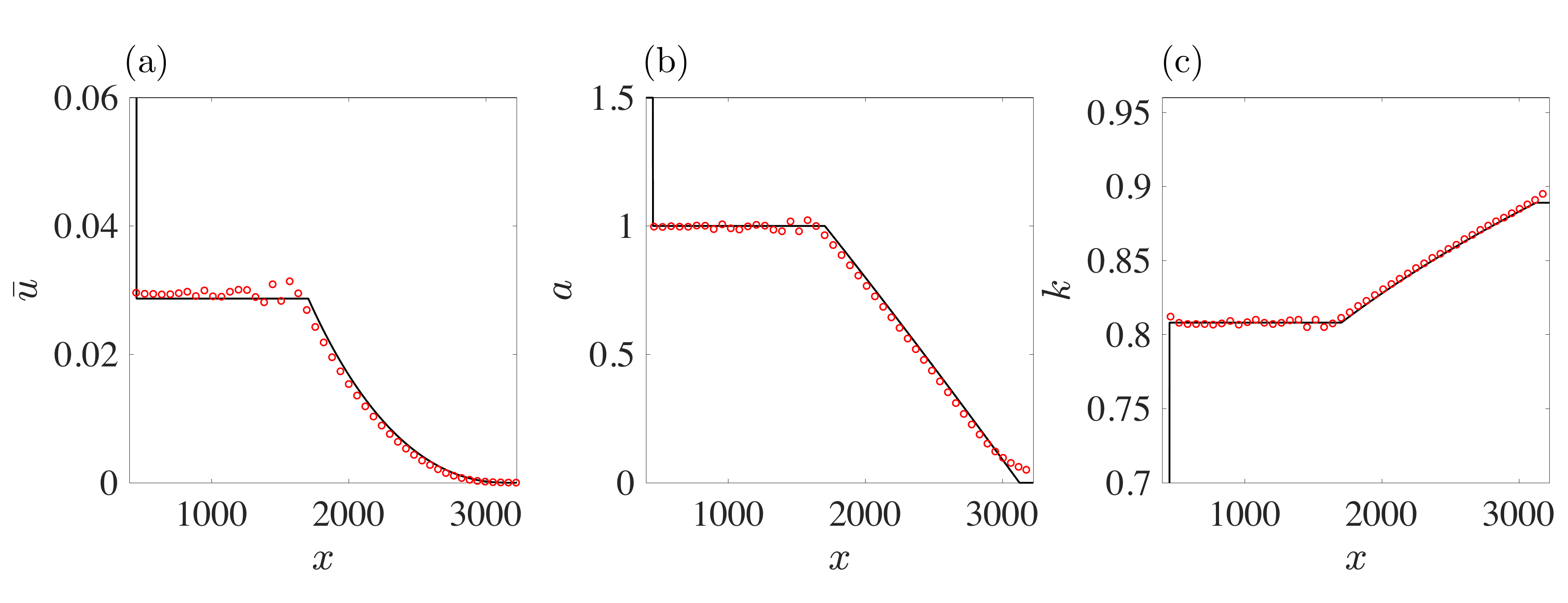}
  \caption{Extracted modulation parameters from numerical simulations
    of \eqref{eq:kdv5_riemann} (circles) with the shock-rarefaction
    solution of the Whitham modulation equations at $t = 1000$ (solid
    curves).}
\label{fig:tdsw_numerics_comparison}
\end{figure}

\subsection{Traveling waves in gravity-capillary water waves}
\label{sec:tdsw-constr-kawah}
In this section, we consider two models of gravity-capillary water
waves that incorporate higher order dispersive effects.  We recall
the linear dispersion relation for right-moving water waves on a
constant depth, normalized to $1$, with surface tension is given in
nodimensional form by
\begin{align*}
  \omega(k) &= \sqrt{k(1 + Bk^2) \tanh(k)}, \\
            & \sim k  + \frac{1}{6}\left(B - \frac{1}{3}\right)k^3 +
              \frac{1}{360}\left(19 - 30B - 45B^2\right)k^5 +
              \mathcal{O}(k^7).  
\end{align*}
In this section, we are motivated by the physical problem in which $B$
is sufficiently close to $1/3$, so that $B-1/3 \sim k^2$ when
$k \ll 1$ so that third and fifth order dispersion are in competition.

\subsubsection{Kawahara equation}

A weakly nonlinear, long wave approximation for water waves with Bond
number $B$ sufficiently close to but less than $1/3$ is the Kawahara equation
\eqref{eq:kawahara}, where $\alpha = 1$\cite{hunter_existence_1988}. Previous numerical simulations of the
Riemann problem \eqref{eq:riemann_data} for eq.~\eqref{eq:kawahara}
show that for sufficiently large values of the initial jump $\Delta$,
the result is a TDSW \cite{sprenger_shock_2017}. Motivated by these
numerical experiments, we now construct traveling wave solutions of
the underlying PDE by computing the Whitham shock loci of the
modulation system using the weakly nonlinear Stokes approximation.

The Kawahara Whitham modulation equations are obtained by averaging the first two conservation laws 
%
\begin{align*}
  (\ub)_t + \left( \frac{1}{2}\ol{\varphi^2}  \right)_x & = 0 , \\
  \left(\frac{1}{2}\ol{\varphi^2}\right)_t + \left( \frac{1}{3}\ol{\varphi^3} - \frac{3}{2} k^2\ol{\varphi_\theta^2}   + \frac{5}{2}k^4\ol{\varphi_{\theta\theta}^2}\right)_x & = 0, \\
  k_t + (ck)_x & = 0,
\end{align*}
where $\varphi = \varphi(\theta)$ is a periodic, traveling wave
solution to \eqref{eq:kawahara}.  Based on numerical computations in
\cite{sprenger_shock_2017}, we expect that a heteroclinic TW solution
exists and connects two disparate far-field waves in which the
leftmost state is a solitary wave with mean $\ub_-$. The jump
conditions are
\begin{align}
  -V\ub_-+ \left( \frac{\ub_-^2}{2}  - \frac{1}{2}\ol{\varphi_+^2}\right)& = 0, \label{eq:kaw_jump1}\\
  -V\left(\ub_-^2 - \ol{\varphi_+^2}\right) + \left(\frac{\ub_-^3}{3} - \frac{1}{3}\ol{\varphi^3_+} + \frac{3}{2}k_+^2\ol{\varphi_{+,\theta}^2}- \frac{5}{2}k_+^4\ol{\varphi_{+,\theta\theta}^2}\right) & = 0, \label{eq:kaw_jump3}\\
  -V k_+ + \omega_+ & = 0.
\end{align}
We use the leading, periodic solution in the form of a Stokes wave
\cite{sprenger_shock_2017}
\begin{align}\label{eq:kaw_stokes}
  \varphi = \ub + a \cos \theta  + \frac{a^2}{12k^2(1-5k^2)} + \ldots,
\end{align}
with approximate phase velocity
\begin{align}
  c = \ub - k^2 + k^4 + \frac{a^2}{24k^2(1-5k^2)} + \ldots .
\end{align}
Inserting the Stokes expansion \eqref{eq:kaw_stokes} into the jump
conditions \eqref{eq:kaw_jump1}--\eqref{eq:kaw_jump3}, we then solve
the algebraic system numerically for a fixed value of $\ub_- = \Delta$
and the ensuing wave parameters are then compared to numerical
solutions of the Riemann problem \eqref{eq:riemann_data}. In Figure
\ref{fig:kaw_tdsw}, we compare the periodic wave parameter values from
the jump conditions using the Stokes approximation
\eqref{eq:kaw_stokes} and those obtained from numerical computations of heteroclinic traveling waves using Matlab's bvp5c. Details of the numerical method

 Since we are using the Stokes approximation for the leading periodic wave, we only expect
agreement with numerical simulations for sufficiently small wave
amplitude. However, the results depicted in Fig.~\ref{fig:kaw_tdsw}
demonstrate that even for quite large values of $a_+$, the Stokes
approximation remains very accurate. Notice further that the numerical
computations of traveling waves, whose parameters are depicted by the
red dots, cease for $\Delta \lesssim 0.58$. Here, numerical fixed
point computations of traveling waves do not converge, yet the Whitham
shock locus continues well past this point. We interpret these
computations as evidence of \textit{inadmissible} Whitham shocks for
the Kawahara equation.  The mere existence of a Whitham shock is not a
sufficient condition to guarantee the existence of a traveling wave
solution.

\begin{figure}[h!]
\begin{center}
\includegraphics[scale=0.4]{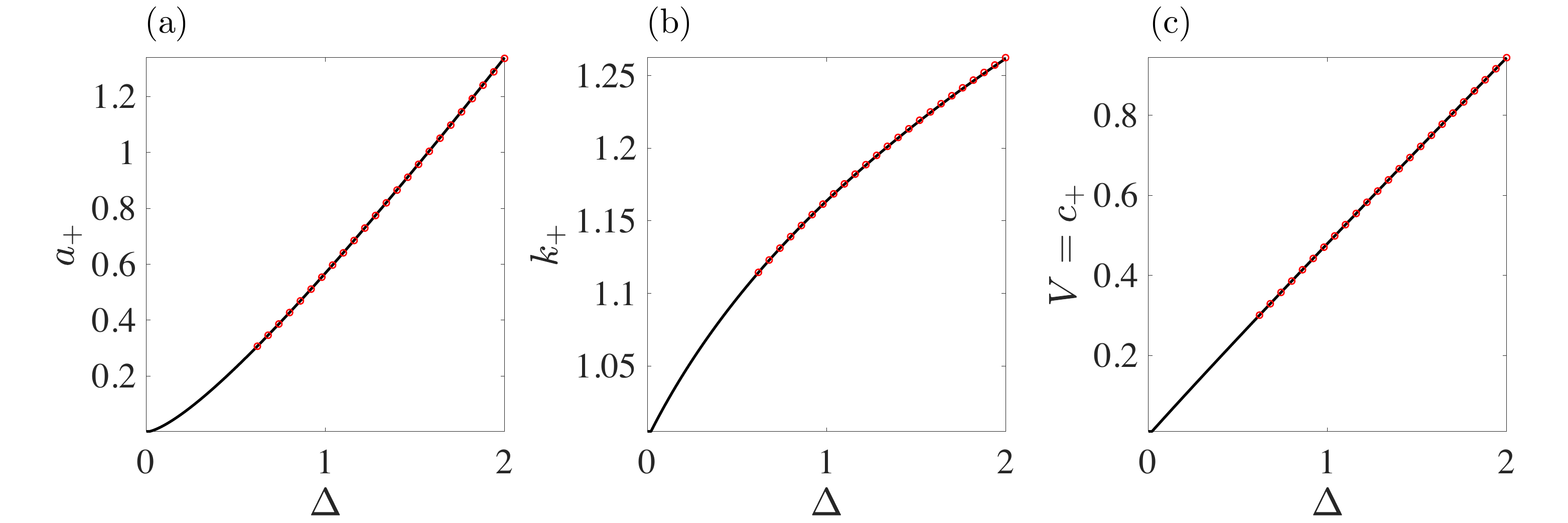}
\caption{(a)-(c)Traveling wave parameter values, $a_+$, $k_+$ and $V$,
  respectively from computations of the jump conditions
  \eqref{eq:kaw_jump1}--\eqref{eq:kaw_jump3} (solid, black curves)
  compared against numerical computations of traveling waves from
  \cite{sprenger_shock_2017} (red
  circles). 
}
\label{fig:kaw_tdsw}
\end{center}
\end{figure}

\subsubsection{Gravity-capillary Whitham equation}

In 1967, Whitham proposed a method to capture the full dispersive
properties of a physical system by using a convolution kernel so that
the linear dispersion relation of a weakly nonlinear model is
prescribed \cite{whitham_variational_1967}. The so-called Whitham
equation takes the form
\begin{align}\label{eq:whitham_eq}
  u_t + uu_x + \mathcal{K}*u_x = 0,
\end{align}
where $\mathcal{K} = \mathcal{F}^{-1}c(k)$, where $c(k)$ is the phase
velocity. For applications to right-moving gravity-capillary water
waves, we take $c(k) = \sqrt{\frac{(1 + Bk^2)\tanh k}{k}}$.

Solutions similar to those appearing in the KdV5 equation
\eqref{eq:kdv5} and the Kawahara equation \eqref{eq:kawahara} are
expected to also persist for the Whitham equation when the Bond number
$B$, is sufficiently close to $1/3$ and third and fifth order
dispersion are in balance. We revisit our original motivation for
investigating traveling wave solutions to the KdV5 equation
\eqref{eq:kdv5} by conducting direct numerical simulations of the
Riemann problem \eqref{eq:riemann_data} for \eqref{eq:whitham_eq} with
$\Delta = 0.25$. The results of the simulations at $t = 2000$ are
shown in Figure \ref{fig:whitham_eq_DSW} for different values of the
Bond number, $B$. The simulations indicate that for $B$ near $1/3$
(e.g., $0.3 \le B \le 0.375$), where third order dispersion is weak,
familiar TDSW-looking dynamics emerge
(cf.~Fig.~\ref{fig:TDSW_intro}). The trailing edge resembles a nearly
uniform traveling wave and the leading portion resembles a partial DSW
terminating into the intermediate nonlinear wavetrain. These numerical
simulations suggest that the Whitham shocks described in this
manuscript are, in fact, quite general. With that being said, the
numerical simulations depicted in Figure \ref{fig:whitham_eq_DSW} for
$B = 0.4$ and $B = 0.425$ no longer resemble TDSWs, instead they look
like classical, convex DSWs. In this regime, lower order dispersion
dominates, preventing the higher order dispersive balance that appears
to be needed for the existence of admissible Whitham shocks. A future
avenue of research will be to understand the transition from solutions
discussed throughout this manuscript to the classical convex, KdV-type
DSW solutions arising from step initial data.

\begin{figure}[h!]
\begin{center}
\includegraphics[scale=0.3]{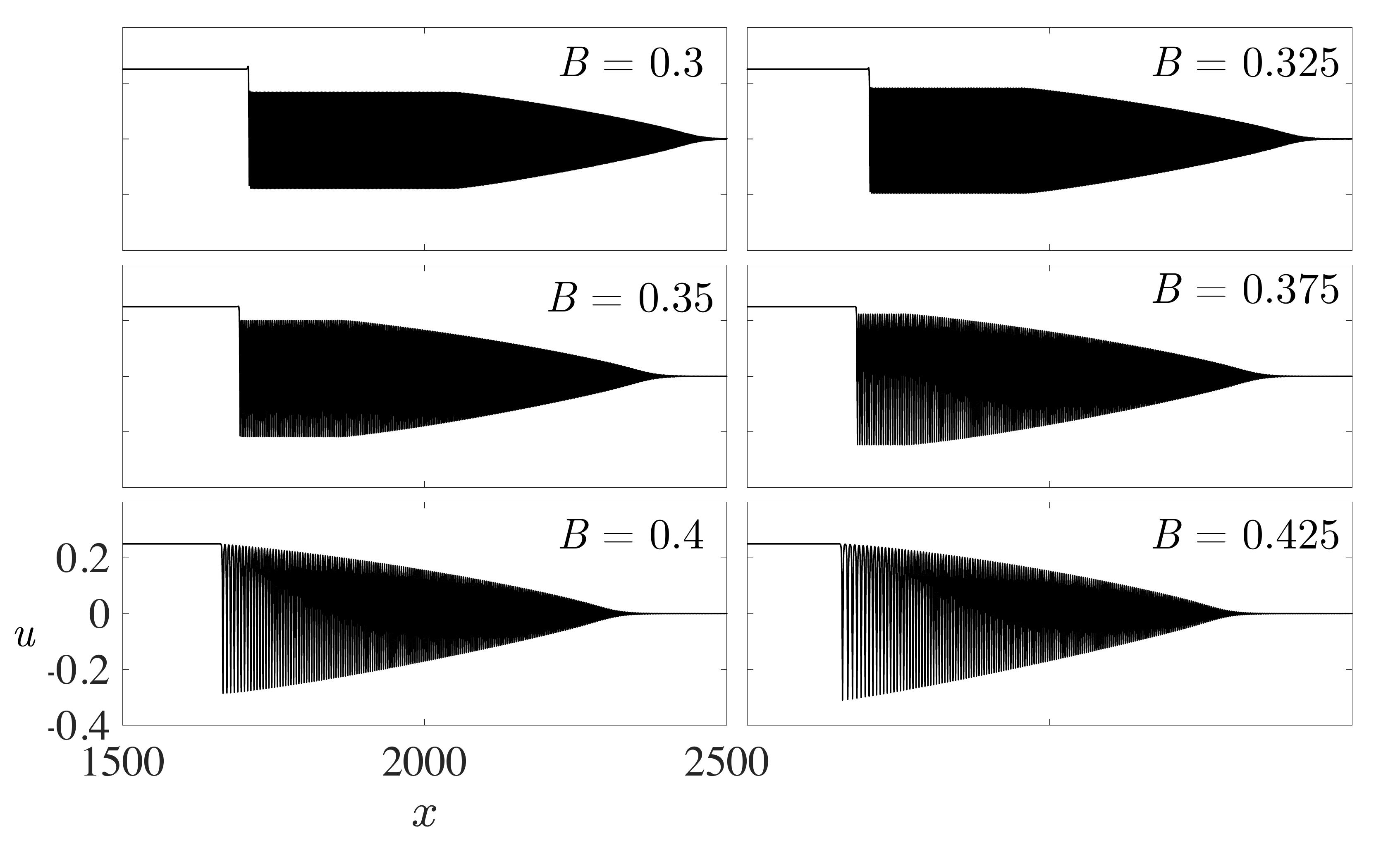}
\caption{Solutions arising from step initial data \eqref{eq:IC} for
  the Whitham equation \eqref{eq:whitham_eq} matching the phase
  velocity, $c(k)$ for linear gravity-capillary water waves. The
  initial step has an amplitude of $\Delta = 0.25$ at $x = 0$ and the
  simulations are shown on the same axes at $t = 2000$ for various
  Bond numbers, $B$, specified in each subfigure.
}\label{fig:whitham_eq_DSW}
\end{center}
\end{figure}

\section{Conclusions}

In this manuscript, we considered the fifth order KdV equation
\eqref{eq:kdv5} and its novel traveling wave solutions. The
heteroclinic TW solutions studied here consist of two periodic waves
that co-propagate and transition between one another on a length scale
commensurate with the length scale of a single wave period. Utilizing
the framework of Whitham modulation theory that describes the zero
dispersion limit of nonlinear oscillatory wavetrains via the Whitham
system of conservation laws, we establish that these heteroclinic
traveling wave solutions limit to discontinuous shock solutions of the
Whitham modulation equations, i.e., \emph{Whitham shocks}.
Within this framework, we prove that a heteroclinic traveling wave
solution of the KdV5 equation \eqref{eq:kdv5} connecting disparate
periodic waves necessarily satisfies the Rankine-Hugoniot jump
conditions of the Whitham equations in conservative form. The jump
conditions reveal \textit{admissible} Whitham shock loci---far-field
periodic wave parameters that are then used as candidate traveling
wave solutions of the governing PDE.

Numerical computations of the characteristic velocities in tandem with
the novel traveling wave profiles allow us to specify the
characteristic structure of Whitham shocks. These computations reveal
that, but for a negligibly small portion of parameter space,
admissible Whitham shocks are undercompressive, i.e., each
characteristic family passes through the shock and therefore violates
the Lax entropy condition.  
The characteristic velocities of the modulation equations also reveal
the stability of heteroclinic traveling wave solutions. Since the
hyperbolicity of the Whitham modulation equations is a necessary
condition for the stability of periodic solutions
\cite{benzoni-gavage_slow_2014}, the stability of traveling wave
solutions considered in this manuscript are determined by the
stability of the periodic waves that constitute them, which can be
readily checked by examining the Whitham shock velocities.

%

A prominent contributing feature to the emergence of novel
heteroclinic traveling waves and admissible Whitham shocks found in
this manuscript is the presence of higher order dispersive
effects. For example, in scalar problems exhibiting nonlocal
dispersion, such as the Whitham equation in
Sec.~\ref{sec:applications}, we demonstrate numerical evidence of the
ubiquity of Whitham shocks in other
systems. 
To further analytically study these heteroclinic traveling waves
connecting two far-field periodic waves, one promising route is to
investigate integrable higher order nonlinear, dispersive systems
\cite{ablowitz_solitons_1981} so that heteroclinic TWs may be
understood in the context of the inverse scattering transform. For
instance, the mathematical structure underlying these traveling waves
could be eludicated though a detailed study of the Lax equation
\cite{wazwaz_partial_2009,lax_integrals_1968}--an integrable equation
with higher order dispersive terms.

%
%

Of related importance is the extensive literature describing
spontaneous, localized pattern formation in \emph{dissipative}
systems, where the prototypical example is the Swift-Hohenberg
equation that leads to a fourth order ODE for stationary solutions
that resemble the TW ODE considered here
\cite{knobloch_spatially_2008,burke_homoclinic_2007,burke_localized_2006,sandstede_stability_1998,beck_snakes_2009,doelman_dynamics_2009}. Here,
stationary solutions of the governing equation were computed in which
an equilibrium state spontaneously transitions to a localized periodic
state and back to equilibrium, the stationary, dissipative analog of
the homoclinic TWs limiting to double Whitham shocks computed
here. The results have since been extended to consider steady
solutions in which a large, localized periodic pattern persists on a
small amplitude oscillatory background
\cite{knobloch_multitude_2019}. A promising avenue for further study
is the consideration of modulation dynamics in the Swift-Hohenburg
equation. Of course, the ensuing dynamics and stability of these
solutions will differ from the dynamics studied here because the
regularizing mechanisms in each case---dissipation or dispersion
respectively---are wholly different. From this point of view, TW
solutions occurring in higher order dispersive systems could serve as
a bridge between conservative Hamiltonian systems and dissipative
pattern forming systems.
\label{sec:conclusion}

\ack 
This work was supported by NSF grant DMS-1816934.  The authors gratefully acknowledge valuable discussions 
with Michael Shearer and Noel Smyth.

\section*{Appendix A}%
In this appendix, we outline numerical methods used throughout this manuscript. 

First, we describe the method used to compute periodic traveling wave solutions to \eqref{eq:kdv5}.  We seek a one parameter family of solutions, $\tilde{\varphi}$ parameterized by an arbitrary wavenumber $\tilde{k}$ and fixed amplitude $\tilde{a} = 1$ and mean $\tilde{\ub} = 0$. For clarity in notation we will continue to denote functions, variables, and parameters associated with the one parameter family with tildes ``$\sim$". The one parameter family of solutions is computed through a pseudospectral Fourier projection of the differential equation \eqref{eq:TW_ode} with $A = 0$ without loss of generality and where $c = \tilde{c}$ is an eigenvalue of the differential operator and used as a continuation variable.

On the interval $\xi \in [-\pi,\pi)$ with periodic boundary conditions the periodic wave is approximated by
\begin{equation}
U_N(\xi) = \sum\limits_{n = -N}^{N-1} a_n e^{ink\xi}
\end{equation}
which yields a nonlinear system of algebraic equations for the Fourier coefficients $a_n$. The nonlinear term is numerically computed in physical space while the linear terms are computed directly via the fast Fourier transform. The resulting nonlinear system is then solved via Matlab's fsolve function. The number of Fourier modes, $N$, is chosen sufficiently large so that Fourier modes decay to $O(10^{-16})$. The nonlinear solution can then be rescaled using the symmetries \eqref{eq:19}, \eqref{eq:20} so that the solution is of unit amplitude and zero mean. The solution library following scaling to unit amplitude and zero mean consists of approximately $6000$ periodic waves non-uniformly spaced on the grid $\tilde{k} \in [0.002, 1.25]$. 

We compute the KdV5-Whtiham modulation equations \eqref{eq:whitham_noncons} by numerically averaging the periodic solutions of the KdV5 equation \eqref{eq:kdv5} over their periods. The averaged integrals, $\tilde{I}_j$ and $\tilde{J}_2$ in \eqref{eq:scaled_ints} are computed using the spectrally accurate trapezoidal rule. We use a finite difference scheme on the unequally spaced grid in $\tilde{k}$ to compute the appropriate gradients of $\tilde{I}_j$ and $\tilde{J}_2$ which are then used to build the matrix $\mathcal{A}$ appearing in Eq. \eqref{eq:whitham_noncons}. The  eigenvalues and eigenvectors of $\mathcal{A}$ are determined directly using Matlab's eig function. Next, we compute the quantities $\mu_j$ in Eq. \eqref{eq:genNL} on the grid $\tilde{k}$ using the aforementioned finite differencing scheme. 

The numerical computation of heteroclinic traveling waves is
implemented using the iterative Newton-conjugate gradient method
\cite{yang_newton_2009,yang_nonlinear_2010}. The numerical
implementation relies upon periodic boundary conditions and projection
onto a Fourier basis. By considering the even reflection of the
Whitham shock at the peak of the right oscillatory wave $\varphi_+$,
the resulting extended structure then satisfies the periodic boundary
condition requirement for the numerical method. An example of the
initial guess for the iterative solver is shown in
Fig. \ref{fig:init_guess}, where intermediate periodic wave wavenumber
is determined by the jump conditions. Outside the oscillatory region
in the initial guess, the solution abruptly transitions to the known
far-field constant value.

\begin{figure}[h!]
\begin{center}
\includegraphics[scale=0.4]{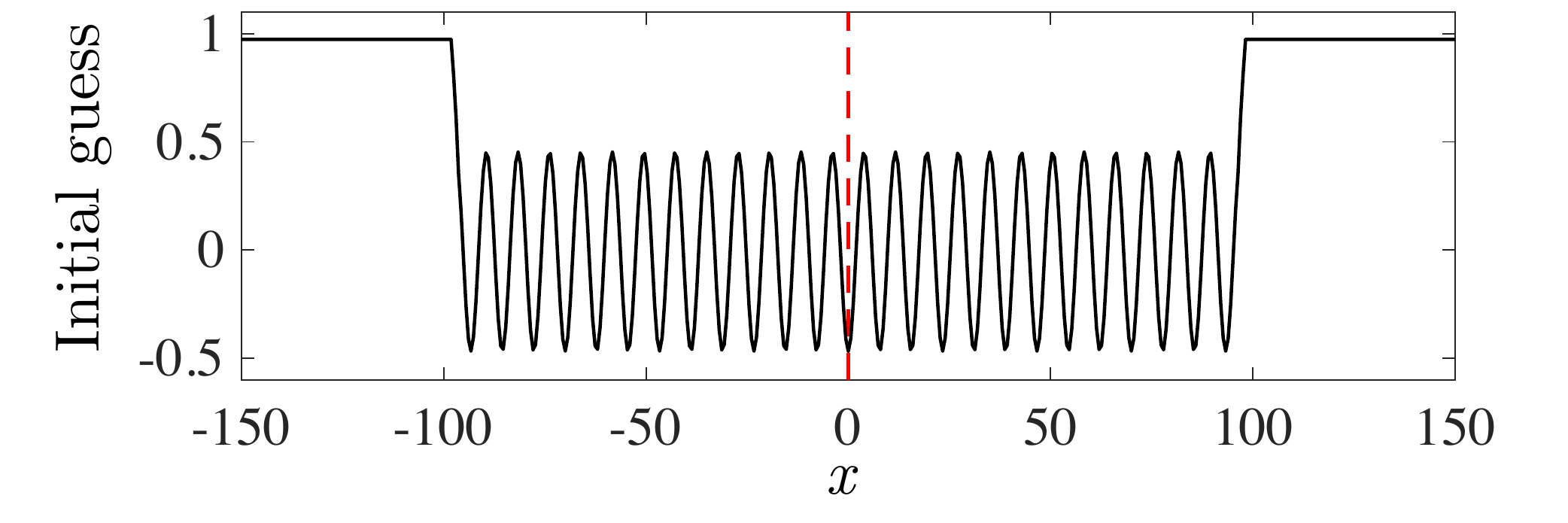}
\caption{Depiction of initial guess given to numerical method with an even reflection across the origin, identified by the vertical, dashed line.}
\label{fig:init_guess}
\end{center}
\end{figure}

Homoclinic traveling waves corresponding to double Whitham shock
solutions of the modulation equations are computed using the fifth
order collocation method bvp5c in Matlab with periodic boundary
conditions imposed on the ODE \eqref{eq:TW_ode}. The solutions are
used as initial contions to KdV5 \eqref{eq:kdv5}, which is numerically
evolved with a pseudospectral spatial discretization via the FFT and
an integrating factor fourth order Runge-Kutta method
\cite{trefethen_spectral_2000} with time step in the range
$\Delta t \in [10^{-4},10^{-3}]$.

\section*{References}

\providecommand{\newblock}{}


\end{document}